\documentclass[a4paper,12pt]{article}

\usepackage{jheppub} 
\usepackage{hyperref}
\usepackage[T1]{fontenc} 

\usepackage{float, array, xspace, amscd, amsthm, amsmath, amssymb}
\usepackage{fancyhdr}
\usepackage{longtable}
\usepackage{tikz-cd,tikz}
\usetikzlibrary{matrix}
\usepackage{comment}
\usepackage{slashed}
\usepackage{mathrsfs}
\usepackage{cancel}
\allowdisplaybreaks

\usepackage[normalem]{ulem}

\theoremstyle{definition}

\newtheorem{lemma}{Lemma}
\newtheorem{proposition}{Proposition}

\newcommand{\nn}{\nonumber}
\newcommand{\dd}{{\rm d}}
\newcommand{\w}{\wedge}
\newcommand{\End}{{\rm End}}

\newcommand{\tr}{{\rm tr}}

\newcommand{\ID}{\mathbb{D}}

\newcommand{\IM}{\mathbb{M}}

\newcommand{\1}{\mathbf{1}}
\newcommand{\2}{\mathbf{2}}

\newcommand{\4}{\mathbf{4}}

\newcommand{\7}{\mathbf{7}}

\newcommand{\mfg}{\mathfrak{g}}

\newcommand{\cB}{\mathcal{B}}

\newcommand{\cD}{\mathcal{D}}
\newcommand{\cE}{\mathcal{E}}

\newcommand{\cL}{\mathcal{L}}

\newcommand{\cN}{\mathcal{N}}

\newcommand{\cQ}{\mathcal{Q}}

\newcommand{\cZ}{\mathcal{Z}}

\newcommand{\newD}{{\check {\mathbb D}}}

\newcommand{\Tmap}
{\mathcal{T}}
\newcommand{\Smap}
{\mathcal{S}}

\newcommand{\be}{\begin{equation}}
	\newcommand{\ee}{\end{equation}}
\def\bea#1\eea{\begin{align}#1\end{align}}

\title{Quantum aspects of heterotic $G_2$ systems}

\author[a]{Xenia de la Ossa,}
\author[b,c]{Magdalena Larfors,}
\author[c]{Matthew~Magill,}
\author[d]{Eirik~E.~Svanes.}

\affiliation[a]{Mathematical Institute, Oxford University}
\affiliation[b]{Department of Physics and Astronomy, Uppsala University}
\affiliation[c]{Department of Mathematics, Uppsala University}
\affiliation[d]{Department of Mathematics and Physics, University of Stavanger}

\preprint{UUITP-33/24}

\emailAdd{delaossa@maths.ox.ac.uk, magdalena.larfors@physics.uu.se, matthew.magill@math.uu.se,
eirik.svanes@uis.no}

\abstract{ 
Compactifications of the heterotic string, to first order in the $\alpha’$ expansion, on manifolds  with integrable $G_2$ structure give rise to three-dimensional ${\cal N} = 1$ supergravity theories that admit Minkowski and AdS ground states. 
As shown in \cite{delaOssa:2019cci}, such vacua correspond to critical loci of a real superpotential
 $W$. We perform a  perturbative study around a supersymmetric vacuum of the theory, which confirms that the first order  variation of the superpotential, $\delta W$, reproduces the BPS conditions for the system, and furthermore shows that $\delta^2 W=0$ gives the equations for infinitesimal moduli. This allows us to identify a nilpotent differential, and a symplectic pairing, which we use to construct a bicomplex, or a double complex, for the heterotic $G_2$ system. Using this complex, we determine infinitesimal moduli and their obstructions in terms of related cohomology groups.  Finally, by interpreting $\delta^2 W$ as an action, we compute the one-loop partition function of the heterotic $G_2$ system and show it can be decomposed into a product of one-loop partition functions of Abelian and non-Abelian instanton gauge theories.}

\date{\today}

\begin{document}

\maketitle

\newpage

\section{Introduction}

There has been a lot of recent interest in both the physics and mathematics literature in studying coupled systems of gauge theory and gravity (geometry). The heterotic string provides a fertile ground for this endeavour. Heterotic strings come with a gauge bundle, whose connection should, for supersymmetric solutions,  satisfy an instanton condition that depends on the geometry.   
Moreover, the coupling of the gauge sector with the geometry through the heterotic Bianchi identity forces the geometry to be generically torsional - a further layer of complication which has made understanding mathematical properties of these solutions, such as the moduli problem, more challenging.

In this paper we study the moduli of heterotic string compactifications on manifolds endowed with a $G_2$ structure, leading to a system of equations called the heterotic $G_2$ system \cite{delaOssa:2017pqy}.
This system corresponds to minimally supersymmetric, three-dimensional solutions of the heterotic string theory to first order in the $\alpha'$ expansion, as explored in e.g.  \cite{Gunaydin:1995ku,Friedrich:2001yp,Gauntlett:2002sc,Gurrieri:2004dt,Micu:2004tz,Gran:2005wf,Gran:2007kh,Lukas:2010mf,Gray:2012md,delaOssa:2014lma,Gran:2016zxk}. 
An interesting aspect of these systems is that they admit an 
${\cal N}=1$ three  dimensional effective field theory on both anti de Sitter (AdS$_3$) and Minkowski ($\IM_3$) spacetimes. The curvature radius of the AdS$_3$ space is related to a (scalar) torsion class which vanishes only in the case of $\IM_3$ \cite{delaOssa:2019cci}. In contrast, ${\cal N}=1$ compactifications of the heterotic strings to four dimensions, which are known as Hull-Strominger systems \cite{Hull:1986kz,Strominger:1986uh}, do not admit AdS$_4$ solutions.

In the last decade, progress has been made towards a better understanding of heterotic systems.  For Hull--Strominger systems, the infinitesimal moduli were parameterised as the first cohomology of a double extension sequence \cite{Melnikov:2011ez, Anderson:2014xha, delaOssa:2014cia}. The moduli problem has been shown to be elliptic, and connected to moduli spaces of solutions of suitable Killing spinor equations on a (holomorphic) Courant algebroid \cite{Garcia-Fernandez:2015hja}. Finite deformations have been explored \cite{Ashmore:2018ybe}, alongside the moduli space geometry, moment maps, stability, flow equations, and further connections with generalised geometry, see e.g. \cite{Candelas:2018lib,  Garcia-Fernandez:2018ypt, Ashmore:2019rkx, Garcia-Fernandez:2020awc, McOrist:2021dnd,   Garcia-Fernandez:2024ypl, Kupka:2024rvl} and references therein.

For heterotic $G_2$ systems,  the infinitesimal moduli were shown to be parametrised by the first cohomology of an operator $\check{\cal D}$, which operates on a bundle $\cal Q$ which couples geometric and gauge degrees of freedom \cite{delaOssa:2016ivz, delaOssa:2017pqy, delaOssa:2017gjq}. For $\check{\cal D}$ to be nilpotent, and thus define a cohomology, it is important that its curvature satisfies the $G_2$ instanton condition. When this is the case, the related $\check{\cal D}$ complex is elliptic \cite{delaOssa:2016ivz, delaOssa:2017pqy}. An interesting consequence of this is that the $\cal Q$ bundle of the heterotic $G_2$ system satisfies the same instanton condition as does the heterotic gauge bundle, and, by dimensional oxidation, a similar result can be shown for the Hull--Strominger system \cite{delaOssa:2017gjq}.  This is an example of what has come to be known as ``instantons for instantons'' and has been explored in depth in \cite{Silva:2024fvl}. Further investigations of heterotic $G_2$ geometry and its moduli, and their connections to generalised geometry has since been explored, both from a mathematical and physical standpoint, see \cite{ Clarke:2016qtg, Fiset:2017auc, Clarke:2020erl, McOrist:2024ivz}. Explicit local and global solutions to the system have also been constructed, e.g. \cite{Fernandez:2008wla, delBarco:2020ddt, Klaput:2011mz, Lotay:2021eog, delaOssa:2021qlt, delaOssa:2021cgd,2023arXiv230607128F, Galdeano:2024fsc}. In spite of this progress, the mathematical structures governing these $G_2$ geometries are more mysterious than their six-dimensional $SU(3)$ counterparts, and tools for understanding them are sorely needed.

As history has shown, physics can provide useful insights into the understanding geometric structures. Examples of this include mirror symmetry \cite{Candelas:1990rm}, and the use of topological field theory to study geometric invariants, such as Chern--Simons theory and the Jones polynomial of knots \cite{Witten:1988hf}, or holomorphic Chern--Simons theory and Donaldson--Thomas invariants \cite{thomas1997gauge, donaldson1998gauge}. In the case of heterotic solutions, physics provides us with a useful functional, the superpotential, whose critical points are exactly the solutions of the heterotic BPS equations \cite{delaOssa:2019cci,delaOssa:2015maa}. In the six-dimensional case, the superpotential was essential in the derivation of the Maurer--Cartan equation of the $L_\infty$ algebra governing finite deformations \cite{Ashmore:2018ybe}. It is natural to ask if similar progress can be made in the $G_2$ case. However, the six-dimensional $SU(3)$ case has an advantage over the $G_2$ solutions: it preserves more supersymmetry, and as a result, the superpotential is required to be a  holomorphic functional of the fields.  This provides a useful parametrisation of the moduli which can be exploited in understanding higher order deformations. For the $G_2$ case, an analogous  parametrisation is not yet known. Yet, we may explore the deformations to quadratic order, and we will do this in the present paper.

Following this perturbative strategy, we will, in this paper, study classical and quantum aspects of  the three-dimensional ${\cal N} = 1$ supergravity with superpotential $W$. We will work perturbatively around a supersymmetric vacuum of the theory. To first order, the variation of the superpotential reproduces the BPS conditions for the system, and, as we will see, the second order variations give the equations for moduli. Moreover, the analysis allows us to construct a bicomplex, or a double complex, which also encodes the gauge symmetry of the theory. Using this double complex, we can determine obstructions to infinitesimal deformations, and make remarks about finite deformations.

Finally, by interpreting the second order variation of the superpotential as an action, we can explore quantum aspects of the system. 
This direct computation of the quadratic action has two advantages, compared to the recent generalised geometry analysis of Refs. \cite{Kupka:2024rvl}. Firstly, it becomes clear that the heterotic $G_2$ moduli problem is governed by a double complex, as was suggested in \cite{Kupka:2024rvl}. Secondly, we will see that upon a field redefinition, the total complex can be diagonalised into three separate complexes. This makes it easier to compute and express the one-loop partition function in terms of determinants of Laplacians on more familiar bundles, as has been done in \cite{Ashmore:2023vji} in the case of type II compactifications on $SU(3)$ structure manifolds. We will compute the absolute value of the one-loop partition function here, and leave the study of its phase and $\eta$-invariant, and the extent to which it defines topological invariants, to future work.

The paper is organised as follows. In section \ref{sec:dW2} we introduce the heterotic $G_2$ superpotential, compute its first and second order variations and study its relevance for the moduli problem. We then explore the deformation theory to quadratic order, and define a double complex, or bicomplex, that encodes the action and symmetries of this theory, in section \ref{sec:complex}. With this in hand, we discuss infinitesimal moduli and their obstructions. In section \ref{sec:1loopall} we compute the one-loop partition function of the quadratic theory. Section \ref{sec:concl} contains our conclusions and a discussion of open problems. Two proofs of propositions are provided in  appendix \ref{proof-prop:del2W} and \ref{sec:elliptic}.

 \section{Variations of the superpotential and equations for moduli}
 \label{sec:dW2}

Our aim in this paper is to study
the three-dimensional ${\cal N} = 1$ supergravity with superpotential $W$, which is  associated to the heterotic $G_2$ system. We will work perturbatively around a supersymmetric vacuum of the theory. To first order, the variation of the superpotential reproduces the BPS conditions for the system, and the second order variations give the equations for moduli, which we compute in the latter half of this section. To start, we recapitulate how the superpotential is defined.

\subsection{The superpotential  and the heterotic $G_{2}$ system}

Let $Y$ be a seven dimensional manifold with a $G_{2}$ structure determined by a well defined nowhere vanishing three form, $\varphi$. 
In \cite{delaOssa:2019cci}, the authors obtained the necessary and sufficient constraints for $\cN = 1$ preserving heterotic string compactifications on $Y$ as the critical points of an effective three dimensional superpotential $W$.  This superpotential was obtained directly from dimensional reduction of the ten dimensional heterotic superstring action, keeping terms up to first order in $\alpha'$.
In this paper we use this superpotential to obtain an action whose equations of motion are equivalent to those describing the infinitesimal deformations of the $G_2$ heterotic system.  In this subsection we summarise briefly the results of \cite{delaOssa:2019cci} as they are the starting point of the research presented in this paper.

Consider, then, the superpotential  \cite{delaOssa:2019cci}
\be
W = \frac{1}{2}\, \int_Y e^{-2\phi}\ \left((H + h\, \varphi)\wedge\psi 
+ \frac{1}{2}\, \dd\varphi\wedge\varphi\right)~,\label{eq:spotential}
\ee
where 
$H$ is the flux on the internal 7-manifold $Y$,  $h$ is a constant which determines the AdS$_3$ curvature scale (which is also the three dimensional flux), $\phi$ is the dilaton, and $\psi = *_{\varphi}\varphi$. 
In this work we consider the space $\cal M$ of all quadruples
\be
[(Y,\varphi), (V,A), (TY, \Theta), H]~,\label{eq:quadruple}
\ee
where $(V,A)$ is a vector bundle over $Y$ with connection $A$, $(TY, \Theta)$ is the tangent bundle of $Y$ with connection $\Theta$ and $H$
is the three form flux on $Y$ given by
\be
H = \dd B + \frac{\alpha'}{4} (CS(A) - CS(\Theta))~, \label{eq:flux}
\ee
where $CS(A)$ and $CS(\Theta)$ represent the Chern-Simons forms for the connections $A$ and $\Theta$ respectively. 

The $G_{2}$ structure on $Y$ determined by $\varphi$ has intrinsic torsion classes $\tau_{i}$ determined uniquely by $\varphi$.  Indeed, they are given by the 
unique decomposition of $\dd\varphi$ and $\dd\psi$ into $G_{2}$ irreducible representations
\begin{align}
\dd\varphi &= - i_{T}(\varphi) = \tau_{0}\,\psi + 3 \tau_{1}\wedge\varphi + *\tau_{3}~,\label{eq:dphi}
\\
\dd\psi &= - i_{T}(\psi) = 4\, \tau_{1}\wedge\psi + *\tau_{2}~,\label{eq:dpsi}
\end{align}
where one can write the intrinsic torsion $T$ as a two form valued in $TY$ as 
\be
T^{a} = \widehat T^{a} + \frac{1}{6}\, \tau_{2}^{ab}\, \varphi_{b}~,\label{eq:intT}
\ee
and $\widehat T$ represents the totally antisymmetric part of the torsion
\be
\widehat T = - \frac{1}{6}\, \tau_{0}\, \varphi + \tau_{1}\lrcorner\psi + \tau_{3}~.\label{eq:Tanti}
\ee

A compactification of the heterotic string preserves at least one supersymmetry if the fields in the quadruple \eqref{eq:quadruple} satisfy the conditions derived from 
\[ W= 0~, \qquad \delta W = 0~.\]
As it was shown in \cite{delaOssa:2019cci}, these conditions give precisely the heterotic $G_2$ system of differential equations for the quadruple \eqref{eq:quadruple}.

To compute $\delta W$ we describe the infinitesimal deformations of the $G_{2}$ structure in terms of a vector valued one form $M$ as follows: 
\begin{equation}\label{eq:Mdecomp}
\delta\varphi = i_{M}(\varphi)~,\quad \delta\psi = i_{M}(\psi)
~,\quad M = \frac{1}{7}\, \tr M\, {\bf 1}+ m + \mathring{\Delta} \in \Omega^{1}(Y, TY)\, ,
\end{equation}
where $m\in\Omega^{2}_{\bf 7}(X)$ and $\mathring{\Delta}$ is traceless symmetric.  The variations of the three form flux $H$ are 
\be
\delta H = \dd{\cal B} + \frac{\alpha'}{2}\, \Big( \tr(\alpha\wedge F) - \tr(\kappa\wedge R) \Big)~,\label{eq:deltaH}
\ee
where $\alpha\in \Omega^{1}(Y, {\rm End}(V))$ and 
$\kappa\in \Omega^{1}(Y, {\rm End}(TY))$ and $\dd\cal B$ corresponds to the ($\alpha'$ corrected) variations of  $\dd B$.
The $\alpha'$-corrected variations $\cal B$ of the $B$-field are such that $\dd{\cal B}$ is a well defined three form on $Y$.

The first order variation of the superpotential is given by
\[
\delta W = \frac{1}{2}\,\int_{Y}  e^{-2\phi} \Big( -2\,\delta \phi\,{\mathfrak{w}} +\delta{{\mathfrak{w}}} \Big)~,
\]
where
\[
{\mathfrak{w}} =
 (H + h\, \varphi)\wedge\psi + \frac{1}{2}\, \dd\varphi\wedge\varphi~.
\]
In \cite{delaOssa:2019cci} it was found that
\begin{align}
\begin{split}
\delta{{\mathfrak{w}}} &=  
 i_{\mathring M}(\psi)\wedge (H - \widehat T)
 - ({\cal B} - 2m) \wedge \big( \dd\psi -2\,\dd\phi\wedge\psi\big)
\\[5pt]
& 
+\frac{\alpha'}{2}\,\Big( \tr(\alpha\wedge F) - \tr(\kappa\wedge R) \Big)\wedge \psi 
+ \left(\delta h + \frac{3}{7}\, \tr M\, \Big(h + \frac{1}{3}\,\tau_{0}\Big) \right) \,   \varphi\wedge\psi   ~.
\end{split}
\label{eq:delwbis}
\end{align}
where $\mathring M$ is the traceless part of $M$. { Thus, with reference to equation \eqref{eq:Mdecomp}, $\mathring M = \mathring{\Delta} + m$.  }

Noting that each term in the expression \eqref{eq:delwbis} must vanish separately, we obtain the equations of motion, that is, the heterotic $G_2$ system:
\begin{align}
&\dd(e^{-2\phi}\, \psi) = 0\qquad
 \iff\qquad
\tau_{2}= 0~,\qquad  2\, \tau_{1} = \dd\phi~,
\nn\\[5pt]
& H = \widehat T~,\nn
\\[3pt]
&F\wedge\psi = 0~, \qquad R(\Theta)\wedge\psi = 0 ~, \qquad \label{eq:HEOM}
\\[3pt]
&  h = - \frac{1}{3}\, \tau_{0}~, \qquad \qquad\delta h = 0 ~.\nn
\end{align}
The vanishing of $\tau_2$ is the necessary and sufficient  condition for there to exist a totally antisymmetric torsion (which is then unique) \cite{FriedrichIvanov2001}.  The second equation identifies the torsion with the three form flux $H$. The third line means that the connections $A$ and $\Theta$ are instantons.  Finally, the three dimensional flux $h$ is identified with the torsion class, $\tau_0$, and moreover their values are fixed.

 \bigskip
 
 \subsection{Second order variations of $W$ }  \label{ssec:deldelW} 

We have that
\[
\delta_{2}\delta_{1} W 
= \frac{1}{2}\,\int_{Y}  e^{-2 \phi}   
\big( 2\, (2\,\delta_{1}\phi\,\delta_{2}\phi - \delta_{2}\delta_{1}\phi)\, {{\mathfrak{w}}}
-2\, (\delta_{2}\phi\, \delta_{1}{{\mathfrak{w}}} + \delta_{1}\phi\, \delta_{2}{{\mathfrak{w}}})
+ \delta_{2}\delta_{1}{{\mathfrak{w}}}\big)~.
\]
We will show in this section that
\be
\begin{split}
\frac{1}{2}\left(\delta_{2}\delta_{1} W+ \delta_{1}\delta_{2} W\right) \big|_0 
&=  \frac{1}{4}\, \int_{Y} e^{-2 \phi} \, \left(\delta_{2}\delta_{1}{{\mathfrak{w}}} + \delta_{1}\delta_{2}{{\mathfrak{w}}}\right)\,|_{0} 
\\[8pt]
&= \langle {\cal Z}_{1}, \ID\,{\cal  Z}_{2} \rangle_{\cal E}~,
\end{split}
\label{eq:deldelWonshell}
\ee
where the symbol $|_{0} $ means that after varying we enforce the equations of motion, and 
\[
 {\cal Z}\in {\cal E}^{0} = \Omega^{0}(Y) \oplus \Omega^{1}(Y, {\cal Q}) \oplus \Omega^{2}_{\bf 7}(Y)~.
 \] 
  
  As we will show in section \ref{sec:complex}, this defines a shifted symplectic pairing $\langle\cdot, \cdot \rangle_{\cal E} $, and a nilpotent differential operator $\ID$, which is  self-adjoint with respect to $\langle\cdot, \cdot \rangle_{\cal E} $. {The differential operator is defined such that the equations for infinitesimal moduli of the heterotic system correspond to}
 \[
 \check{\mathbb D}{\cal Z}= 0~,
 \]
 {where $ \check{\mathbb D}$ denotes projection to $G_2$ irreducible representations that will be defined below}. These definitions will allow us, in the ensuing sections, to define a bicomplex that captures the gauge symmetries of the 3-dimensional theory, and quantize the theory. This furthermore requires the present definition of ${\cal Z}$, which is an extension of the infinitesimal deformations defined in \cite{delaOssa:2017pqy}. In that reference, the infinitesimal moduli of the heterotic $G_2$ system were identified with $H^1(Y, {\cal Q})$. As will become evident in section \ref{sec:complex}, this does not suffice to provide a complete description of the gauge degrees of freedom, which we need in order to quantize the theory.
 
 For the second order variation of $W$, we start from  equation \eqref{eq:delwbis} and compute 
 $\delta_{2}\delta_{1}{{\mathfrak{w}}}\,|_{0}$. {For ease of exposition, we relegate most of this analysis to Appendix \ref{proof-prop:del2W}. We also refer the reader to section 2.2.1 of \cite{delaOssa:2017pqy} for definitions of the insertion operator $i_\alpha$, and further explanation of our notation.} The result of the computation is summarised in the following proposition 

\begin{proposition}\label{prop:del2W}
Let
\be
\beta = \frac{1}{2} (\pi_{\bf 7}({\cal B}) - 2 m) ~, \label{eq:beta}
\ee
 \be
y =  M   + \frac{1}{2}\, \pi_{\bf{14}}({\cal B})~, \label{eq:y}
\ee
and
\be
\hat\delta\phi = \delta\phi - \frac{2}{7} \tr M~.
\ee
Then 
 \be
\begin{split}
\frac{1}{2}\, \delta_{2}\delta_{1}{{\mathfrak{w}}}\,|_{0} 
& =   \Big\{ ~~i_{y_{1}}
\Big[
 \dd_{\zeta} y_{2} 
-  \frac{\alpha'}{4}\, \Big( \tr (\alpha_{2} \wedge F) - \tr(\kappa_{2}\wedge R)\Big) 
\Big]
- i_{\mathring M_{1}}(\dd\beta_{2})
\\[5pt]
 &\quad~~ +  
\frac{\alpha'}{4}\,\Big[ \tr\big(\alpha_{1}\wedge (\dd_{A}\alpha_{2} - i_{ M_{2}}(F)\big) 
- \tr\big(\kappa_{1}\wedge (\dd_{\Theta}\kappa_{2}-  i_{ M_{2}}(R) )\big)
 \Big]
 \\[5pt]
&\quad~~ 
- \beta_{1}\wedge\left(
\frac{1}{3}\, \psi\lrcorner\,\left(i_{\dd_{\zeta}{\mathring M}_{2}}(\psi)\right) +  i_{{\mathring M}_{2}} (4\, \tau_{1}) -  2\,\hat\delta_{2}(\dd\phi) \right)
\Big\}\wedge\psi
~, \label{eq:deldelw}
 \end{split}
 \ee
 where  $\zeta$ is the Hull connection given by
 \[
 \dd_{\zeta} V^{a} = \dd V^{a} + \Gamma^{a}{}_{bc}\dd x^{b}\wedge V^{c}~, \quad\forall \,V\in\Omega^{0}(Y, TY)~,
 \]
 and $\Gamma^{a}{}_{bc}$ are related to the symbols of a $G_{2}$ compatible connection by
 \[
 \nabla_{b} V^{a} = \partial_{b} V^{a} + \Gamma^{a}{}_{cb}\,V^{c}~, \quad\forall\, V\in\Omega^{0}(Y, TY)~.
 \]
\end{proposition}

\bigskip


\proof See Appendix \ref{proof-prop:del2W}. \endproof

{We have only defined the Hull connection acting on vector valued $0$-forms. Since this is a metric connection, it is  straight forward to generalize the definition to other tensor valued $p$-forms \cite{delaOssa:2016ivz, delaOssa:2017pqy}.}


Looking back at equation \eqref{eq:deldelWonshell}, we now need to compute the {\it symmetric} second order variation of the superpotential. 
The result of symmetrising the second order variation in proposition \ref{prop:del2W} is summarised in the following proposition. 

\proposition
\label{prop:delW2sym}
 \be
\begin{split}
&\frac{1}{2}\big(\delta_{2}\delta_{1} W + \delta_{1}\delta_{2} W \big) \,|_{0} 
\\[5pt]
& =  \int_{Y} \Big\{ ~~i_{y_{1}}
\Big[
 \dd_{\zeta} y_{2} 
-  \frac{\alpha'}{4}\, \Big( \tr (\alpha_{2} \wedge F) - \tr(\kappa_{2}\wedge R)\Big) 
\Big]
- i_{y_{1}}(\hat\dd\beta_{2})
\\[5pt]
 &\qquad +  
\frac{\alpha'}{4}\,\Big[ \tr\big(\alpha_{1}\wedge (\dd_{A}\alpha_{2} - i_{ M_{2}}(F) )\big) 
- \tr\big(\kappa_{1}\wedge (\dd_{\Theta}\kappa_{2}-  i_{ M_{2}}(R) )\big)
 \Big]
 \\[5pt]
&\qquad
- \beta_{1}\wedge\left[\,
\frac{1}{3}\, \psi\lrcorner\,\left(i_{\dd_{\zeta}{\mathring M}_{2}}(\psi)\right) +  i_{{\mathring M}_{2}} (4\, \tau_{1}) -  \dd(\hat\delta_{2}\phi) \right]
- \hat\delta_{1}\phi\, \dd\beta_{2}
\Big\}\wedge\Psi
~, \label{eq:sym}
 \end{split}
 \ee
 where we recall that $\Psi = e^{-2\phi}\, \psi$ is a closed 4-form, and 
  \[
  \hat \dd \beta = \dd\beta - \pi_{\bf 1}(\dd\beta) = \dd\beta - \check\dd\beta~.
  \]

\proof
Most terms in \eqref{eq:deldelw} are symmetric except those which contain $\beta$.  Indeed, up to total derivatives and supersymmetry conditions:
\begin{align*}
i_{y_{2}} (\dd_{\zeta}y_{1})\wedge\Psi  
&= i_{y_{1}} (\dd_{\zeta}y_{2})\wedge\Psi ~,
\\ 
\tr \big( \alpha_{2}\wedge \dd_{A}\alpha_{1}\big)\wedge\Psi  
&= \tr \big( \alpha_{1}\wedge \dd_{A}\alpha_{2}\big)\wedge\Psi ~,
\\
 - \tr  (\alpha_{2}\wedge i_{y_{1}}(F))\wedge\Psi \, 
 &= - i_{y_{1}}(\tr (\alpha_{2}\wedge F)) \wedge\Psi ~, 
 \\[5pt]
   i_{\mathring M_{2}}\Big(\dd\beta_{1} \Big)\wedge\Psi  
& = \beta_{1}\wedge\left(
 \frac{1}{3}\, \psi\lrcorner\,\left(i_{\dd_{\zeta}\mathring M_{2}}(\psi)\right) +  i_{\mathring M_{2}} (4\, \tau_{1})  \right) \wedge\Psi ~,
\\[5pt]
\beta_{2}\wedge\dd(\hat\delta\phi_{1})\wedge\Psi 
&= - \,\hat\delta\phi_{1}\, \dd\beta_{2}\wedge\Psi~.
 \end{align*}
Here we only  remark on the fourth equation, as it is straightforward to prove the others. To prove this equation have
\[
i_{{\mathring M}_{2}}(\dd\beta_{1})\wedge\Psi 
= i_{{\mathring M}_{2}}(\dd\beta_{1}\wedge\Psi) 
- \dd\beta_{1}\wedge i_{{\mathring M}_{2}}(\Psi)~.
\]
The first term vanishes as $\mathring M$ is traceless.  Then
\[
\begin{split}
i_{{\mathring M}_{2}}(\dd\beta_{1})\wedge\Psi
&= - \dd\big(\beta_{1}\wedge i_{{\mathring M}_{2}}(\Psi)\big)
+ \beta_{1}\wedge \dd i_{{\mathring M}_{2}}(\Psi) 
\\[5pt]
&= \beta_{1}\wedge\big( i_{\dd_{\zeta}{\mathring M}_{2}}(\Psi) + i_{{\mathring M}_{2}}(4\tau_{1})\wedge\Psi
\big)~,
\end{split}
\]
from which the result follows.

Finally, to see how the the second term on the right hand side of equation \eqref{eq:sym} comes about we have
\[
i_{{\mathring M}_{1}}\big(\dd\beta_{2} \big)\wedge\Psi =
\left(i_{M_{1}}\big(\hat\dd\beta_{2} \big)
   + \frac{1}{2}\, i_{\pi_{\bf 14}({\cal B}_{1})}\big(\hat\dd\beta_{2} \big)
   \right)\wedge\Psi
= i_{y_{1}}\big(\hat\dd\beta_{2} \big)\wedge\Psi  
 ~,
  \] 
  where
  \[
  \hat \dd \beta = \dd\beta - \pi_{\bf 1}(\dd\beta) = \dd\beta - \check\dd\beta~.
  \]
  and we have used
  \[
  \hat\dd\beta\wedge\psi = 0~.
  \]

 $\null \hfill\square$ 
\bigskip

  We are now ready to read off from equation \eqref{eq:sym} the set of differential equations that the infinitesimal deformations must satisfy:
\begin{align}
\dd(\hat\delta\phi) - \frac{1}{3}\, \psi\lrcorner\,\big(i_{\dd_{\zeta}\mathring M}(\psi)\big) -  i_{\mathring M} (4\, \tau_{1}) 
  &= 0~. \label{eq:ddeldil}
\\*[5pt]
\pi_{\bf 7}\left[\dd_{\zeta}  y_{a}  
+ \frac{\alpha'}{4}\, \Big( \tr \big(\alpha\wedge F_{a}\big) 
 - \tr \big(\kappa\wedge R_{a}\big)\Big)
    - (\hat\dd\beta)_{a} \right] &= 0~,\label{eq:sysmoduli}
  \\*[5pt]
\pi_{\bf 7}\big(\dd_{A}\alpha - i_{M}(F)\big) &= 0~, \label{eq:Vmoduli}
\\*[5pt]
\pi_{\bf 7}\big(\dd_{\Theta}\kappa -  i_{M}(R)\big) &= 0~, \label{eq:Tmoduli}
 \\*[5pt]
\pi_{\bf 1}(\dd\beta) &= 0~.\label{eq:pi1dbeta}
\end{align}

These equations include and extend the equations for infinitesimal moduli presented by some of the present authors in \cite{delaOssa:2017pqy};  equation \eqref{eq:ddeldil} is equation (3.10) in \cite{delaOssa:2017pqy}, and the equations \eqref{eq:sysmoduli}-\eqref{eq:Tmoduli}
reduce, upon setting $\beta$ to zero, to equations (5.10)-(5.13) in that reference. To comment on these differences, we note that while the constraint $\beta=0$ was an allowed and simplifying assumption in the study of infinitesimal moduli of the heterotic $G_2$ system done in Ref.~\cite{delaOssa:2017pqy}, it is by no means necessary. As we will see below, for the analysis in this paper, $\beta$ and $\delta\hat{\phi}$  will play a prominent role.

\subsection{The pairing $\langle\cdot,\cdot\rangle_{\cal E}$ and the operator $\mathbb D$}
\label{sec:Dops}

 In this subsection, we will repackage the infinitesimal deformations in a manner that allows us to write equation \eqref{eq:sym}
  in the form of equation \eqref{eq:deldelWonshell}, i.e.~to identify the operator $\ID$ and the pairing on  $\cal E$.

 Let
 \be
{\cal Z} = 
  \begin{pmatrix}
    \hat\delta\phi \\
  Z \\
  \beta
  \end{pmatrix}
  ~~\in~~ {\cal E}^{0}~,\qquad 
  Z = \begin{pmatrix}
   y\\
  \alpha \\
  \kappa
  \end{pmatrix} ~~\in ~~\Omega^{1}(Y, {\cal Q})~,
\ee
where  $\beta \in \Omega^{2}_{\bf 7}(Y)$, $ \hat\delta\phi \in \Omega^{0}(Y)$, and
$\cal Q$ is the bundle defined by
\be
{\cal Q} = TY\oplus {\rm End}(V)\oplus {\rm End}(TY)~.
\ee
Let $\cal D$ be the differential operator acting on $\Omega^{p}(Y, {\cal Q})$
defined by \cite{delaOssa:2017pqy}
\be
{\cal D} =  \begin{pmatrix}
    ~ \dd_{\zeta}~&~{\cal F}^{\dagger}~&~{\cal R}^{\dagger}~\\
~{\cal F}~&~\dd_{A}~&~0~\\
~{\cal R}~&~0~&~\,\dd_{\Theta}~~
  \end{pmatrix}~,\label{eq:DonQ}
  \ee
  where
\begin{align}
\begin{split}
{\cal F}~:~ \quad\Omega^{p}(Y, TY) ~~&\longrightarrow~~ \Omega^{p+1}(Y, {\rm End}(V))
\\[5pt]
\qquad\quad y^{a}\quad ~~&\longmapsto\quad {\cal F}(M) = (-1)^{p}\, i_{y}(F)~,
\end{split}\label{eq:calF}
\\[10pt]
\begin{split}
{\cal F}^{\dagger}~:~ \Omega^{p}(Y, {\rm End}(V))~~&\longrightarrow~~ \Omega^{p+1}(Y, TY) 
\\[5pt]
\qquad\qquad \alpha \quad ~~&\longmapsto\quad 
({\cal F}^{\dagger}\alpha)^{a} = - \tr (\alpha\wedge F)^{a}~,
\end{split}\label{eq:calFadj}
\end{align}
with analogous definitions for $\cal R$ and ${\cal R}^{\dagger}$. 
Note that the same labelling  $\cal F$ and ${\cal F}^{\dagger}$ may be used, without ambiguity, for the action of these maps on different $p$-forms.

We also define the following new operations on ${\cal E}^{0}$: 
\begin{align}
\begin{split}
{\tau}~:~ \quad\Omega^{1}(Y, TY) ~~&\longrightarrow~~ \Omega^{1}(Y)
\\[5pt]
\qquad\quad y\quad ~~&\longmapsto\quad \tau(\mathring M) = - i_{\mathring M}(4\tau_{1}) 
- \frac{1}{3}\, \psi\, \lrcorner  i_{\dd_{\zeta}{\mathring M}}(\psi)~,
\end{split}\label{eq:T1}
\\[10pt]
\begin{split}
{\tau}^{\dagger}~:~ \Omega^{2}_{\bf 7}(Y)~~&\longrightarrow~~ \Omega^{2}(Y, TY) 
\\[5pt]
\qquad\qquad \beta \quad ~~&\longmapsto\quad (\tau^{\dagger}\beta)_{a} = - (\hat\dd\beta)_{a}~.
\end{split}\label{eq:T1adj}
\end{align}

With these definitions at hand, equation \eqref{eq:sym} can be written as 
     \begin{equation}
\begin{split}
\frac{1}{2}\big(\delta_{2}\delta_{1} W + \delta_{1}\delta_{2} W \big) \,|_{0} 
& =  \langle Z_{1}, {\cal D} Z_{2}\rangle_{\cal Q} 
+ \int_{Y }\,\Big\{ 
  -\, \hat\delta_{1}\phi\, \dd\beta_{2} 
  \\[8pt]
  & +  \left( i_{y_{1}}\big(\tau^{\dagger}(\beta_{2})\big) 
  + \beta_{1}\wedge \big[\tau(\mathring M_{2}) + \dd\hat\delta_{2}\phi   \big]\right)
  \Big\}\wedge\Psi~,
\end{split}\label{eq:sym2}
  \end{equation} where we recall that $\Psi = e^{-2\phi} \psi$,  and 
$\langle\cdot,\cdot\rangle_{\cal Q}$ is the pairing  on $\cal Q$ (defined  as in \cite{delaOssa:2017pqy})
\be
\langle\cdot,\cdot\rangle_{\cal Q} = \int g_{\cal Q}(\cdot, \cdot)~,
\quad
g_{\cal Q}(\cdot, \cdot)
= 
\begin{pmatrix}
~ i_{\cdot}(\cdot)& 0& ~0\\[4pt]
~ 0 & ~~ \frac{\alpha'}{4}\, {\bf {\tr}}_{V}(\cdot\wedge\cdot) & ~0\\[4pt]
~0 & ~ 0 & - \frac{\alpha'}{4}\, {\bf {\tr}}_{TY}(\cdot\wedge\cdot) ~\,
\end{pmatrix} \wedge\Psi~.\label{eq:innerQ}
\ee

 It is easy to check that  ${\cal F}^{\dagger}$ is indeed the adjoint of $\cal F$ with respect to $\langle\cdot,\cdot\rangle_{\cal Q}\,$.  This follows from the the symmetry of \eqref{eq:sym} discussed in the proof of  proposition \ref{prop:delW2sym}.   Similar considerations apply for the operators $\cal R$ and ${\cal R}^{\dagger}$.   Furthermore, these considerations imply that  $\cal D$ is a self adjoint operator with respect to $\langle\cdot,\cdot\rangle_{\cal Q}\,$, that is, 
 \[
 \langle Z_{1}, {\cal D} Z_{2}\rangle_{\cal Q} = \langle Z_{2}, {\cal D} Z_{1}\rangle_{\cal Q} \:.
 \]

As mentioned before, the  proof is in fact equivalent to the process of finding the symmetric second order variation of the superpotential. 

Finally, we can define a new pairing $\langle\cdot,\cdot\rangle_{\cal E}$, and rewrite equation \eqref{eq:sym2}  as
\begin{equation}
\frac{1}{2}\big(\delta_{2}\delta_{1} W + \delta_{1}\delta_{2} W \big) \,|_{0} 
 = \langle {\cal Z}_{1}, {\mathbb D} {\cal Z}_{2}\rangle_{\cal E}~,
\label{eq:Delta2W}
  \end{equation}
where ${\cal E}^{0} = \Omega^{0}(Y) \oplus \Omega^{1}(Y, {\cal Q}) \oplus \Omega^{2}_{\bf 7}(Y)$,
 and $\mathbb D$ is the differential operator which acts on ${\cal E}^{0}$ given by
\be 
{\mathbb D} =  \begin{pmatrix}
   ~\,  \dd~&~\tau ~&~0~&~0~&~~0~\\
   ~ ~ 0~&~\dd_{\zeta}~&~{\cal F}^{\dagger}&~{\cal R}^{\dagger}~
&~~~ {\tau}^{\dagger}~\\
~~0~& ~{\cal F}~&~\dd_{A}~&~0~&~~\,0~\\
~~0~& ~{\cal R}~&~0~&~~\dd_{\Theta}~~&\quad 0~\\
~~0~&~0~&~0~&~0~& - \dd~
  \end{pmatrix}
  =
  \begin{pmatrix}
   ~\, \dd~&~{\bf \tau} ~&~~0~\\
 ~~ {\bf 0}~&~{\cal D}~&~~~ {\bf\tau}^{\dagger}~\\
  ~~0~&~{\bf 0}~&- \dd~
  \end{pmatrix}~.
  \label{eq:Dcheck}
  \ee
From the computations in section \ref{ssec:deldelW}, in particular from the symmetrisation of the second order variation of the superpotential,  it is clear that
$\mathbb D$ is self adjoint with respect to $\langle\cdot,\cdot\rangle_{\cal E}$ and that $\tau^{\dagger}$
is precisely the adjoint of $\tau$ with respect to this pairing. 

{Clearly, the equations for moduli, \eqref{eq:ddeldil}-\eqref{eq:pi1dbeta}, can be recast in terms of $\mathbb D {\mathcal{Z}}$, provided that the required projections to $G_2$ irreducible representations are enforced. Defining ${\newD}$ to be the projected differential operator, \eqref{eq:ddeldil}-\eqref{eq:pi1dbeta} are equivalent to  }
\[
\check{\mathbb D}{\cal Z} =  \begin{pmatrix}
   ~\, \dd~&~\check{\tau} ~&~~0~\\
 ~~ {\bf 0}~&~\check{\cal D}~&~~~ \check{\tau}^{\dagger}~\\
  ~~0~&~{\bf 0}~&- \check\dd~
  \end{pmatrix}\, {\cal Z}= 0~,
\]
where {the check denotes projection to the $\bf{1}$ and $\bf{7}$ $G_2$ representations, as appropriate, and}
\[
\check\tau(y) = \tau(\mathring M)~,\qquad \check{\tau^{\dagger}}(\beta)  = - \pi_{\bf 7 }(\hat\dd\beta)_{a}~.
\]
{In chapter \ref{sec:defineD}, we will give a definition of $\newD$ acting on other tensor-valued $p$-forms.}

\subsection{Trivial deformations and a first look at infinitesimal moduli}

For completeness, in this subsection we comment briefly on the trivial deformations of the system.  A full discussion is deferred to the next chapter.  

The infinitesimal variations ${\cal Z} \in \cE^0$ are not yet the infintesimal moduli
of the heterotic $G_2$ system. We need to quotient out by the trivial deformations of the system, corresponding to diffeomorphisms and gauge symmetries. As we will see later, these trivial deformations can be shown to be
 \begin{equation}
 \begin{split}
 \label{eq:trivtrans}
 2(\hat\delta\phi)_{triv} &= 
2\, \left( {\cal L}_{V}\phi - \frac{2}{7}\,\tr M_{triv}\right)  \\
  y^{a}_{triv} &= \dd_{\zeta}V^{a} + {\cal F}^{\dagger}(\epsilon)^{a} + (\hat\dd\mu)^{a}~,
  \\
  \alpha_{triv} &= \dd_{A}\epsilon + {\cal F}(V)
  \\
  \kappa_{triv} & =  \dd_{\Theta}\lambda + {\cal R}(V)
  \\
  \beta_{triv} & = \dd\mu
 \end{split}
 \end{equation}
where $V\in \Omega^{0}(Y, TY)$, $\mu \in \Omega^1(Y)$, $\epsilon \in \Omega^{0}(Y, {\rm End}(V))$, $\lambda  \in \Omega^{0}(Y, {\rm End}(TY))$. There is also a gauge-of-gauge symmetry since the $B$ field is a 2-gerbe:
\be \label{eq:trivtrans2}
\mu_{triv} = \dd \nu \, , \; \nu \in \Omega^0(Y) \; .
\ee

The trivial deformations of the heterotic $G_2$ system corresponding to the $\cQ$-valued elements $Z = (y, \alpha, \kappa)^T$ were explained in detail in  Ref.~\cite{delaOssa:2017pqy}, where it was shown that there exists a canonical $G_2$ complex
 \begin{equation}
 \label{eq:qcomplex}     0\rightarrow\Omega^0(\cQ)\rightarrow\Omega^1(\cQ)\rightarrow\Omega^{2}_{\bf 7}(\cQ)\rightarrow\Omega^{3}_{\bf 1}(\cQ)
 \rightarrow 0
 \end{equation}
 which is elliptic. The nilpotent differential  of this complex is the operator $\check\cD$ defined in the previous subsection. Focusing on the $\cQ$-valued elements of the deformations,  it can be shown that the infinitesimal moduli are captured by the cohomology group $H^1(Y,\cQ)$  \cite{delaOssa:2017pqy}.
 
 In the next section, we will write down an analogous complex  that also includes the degrees of freedom related to the dilaton and $B$-field.

\section{Defining the complex}
\label{sec:complex}

In the last section, we showed that the equations satisfied by the infinitesimal deformations for the heterotic $G_2$ system correspond to 
\[
 \check{\mathbb D}{\cal Z}= 0~,
 \]
a constraint which ensures the vanishing of the second order variation of the heterotic $G_2$ superpotential. We also made some observations about the trivial deformations of the system. 

In the rest of this paper, we will interpret this constraint as the equations of motion of an action, which we will write as 
\begin{equation}
    \label{eq:action}
  S= \langle {\cal Z}, \newD \,{\cal  Z} \rangle_{\cal E}  ~.
 \end{equation} 
This will allow us to formalize and extend the observations made for infinitesimal deformations.  To do so, we start in this section, by constructing a  complex 
 \begin{equation}
 \label{eq:bicomplex}
     0\rightarrow\cE^{-2}\rightarrow\cE^{-1}\rightarrow
     \cE^{0}\rightarrow\cE^{1}\rightarrow\cE^{2}\rightarrow\cE^{3}\rightarrow 0
 \end{equation}
where the arrows correspond to action by the operator $\newD$, which extends the operator defined above to act on elements in $\cE^d$, vector spaces of homological degree $d$. In more detail, 
\[
\cE^*=\check{\Omega}^*[2]\oplus\check{\Omega}^*(\cQ)[1]\oplus\check{\Omega}^{4+*} \, .
\]
where the numbers within the brackets denote the shift of the homological degree, { and the checks denote the appropriate projections to $G_2$ representations,\footnote{Although it will not be of relevance for the present analysis, we note that \eqref{eq:bicomplex} is in fact the total complex of a bicomplex, see section \ref{app:bicomplex}.}
e.g. 
\begin{equation*}
\cE^0=(\check{\Omega}^*[2]\oplus\check{\Omega}^*(\cQ)[1]\oplus\check{\Omega}^{4+*}) ^0
=\check{\Omega}^{2}\oplus \check{\Omega}^1(\cQ)\oplus \check{\Omega}^{4}=\Omega^{2}_{\bf 7}\oplus {\Omega}^1(\cQ)\oplus \Omega^{4}_{\bf 1} \; .   
\end{equation*}}

This complex encodes the equations of moduli and symmetries of the heterotic $G_2$ system. 
The aim of this section is to work out the properties of this  complex in detail. To that end, let us first write out the complex in the following form  
\begin{equation}\label{diag:defn1}
	\begin{tikzcd}
		& & \Omega^{4}_{\bf 1}\arrow{r}{\check{\dd}}&\Omega^{5}_{\bf 7}\arrow{r}{\check{\dd}}&\Omega^6\arrow{r}{\check{\dd}}&\Omega^7\\
		& \Omega^0(\cQ)\arrow{r}{\cD}\arrow{ru}{\Tmap^{(-1)}}&\Omega^1(\cQ)\arrow{r}{\Psi\w\cD}\arrow{ru}{\Tmap^{(0)}}&\Omega^6(\cQ)\arrow{r}{\cD}\arrow{ru}{\Tmap^{(1)}}&\Omega^7(\cQ)\\
		\Omega^0\arrow{r}{-\check{\dd}}&\Omega^1\arrow{r}{-\check{\dd}}\arrow{ru}{\Smap^{(-1)}}&\Omega^{2}_{\bf 7}\arrow{r}{-\check{\dd}}\arrow{ru}{\Smap^{(0)}}&\Omega^{3}_{\bf 1}\arrow{ru}{\Smap^{(1)}}
	\end{tikzcd}
\end{equation} 
The presentation of the complex is chosen in order to simplify our analysis below, which possibly makes the connection to the heterotic $G_2$ system, and the $\cal Q$ complex \eqref{eq:qcomplex} a bit opaque.\footnote{An alternative, isomorphic, presentation would be to work with $\check {\cal D}$ differentials in the middle row, as in \eqref{eq:qcomplex}. } 
 Each row of this total complex is an elliptic complex, that are all modelled on the canonical $G_2$ complex introduced in \cite{ReyesCarrionPhd,ReyesCarrion:1998si,fernandez1998dolbeault} and the $\cal Q$ complex of \cite{delaOssa:2017pqy}.  This implies, in particular, that the complex \eqref{diag:defn1} has nilpotent horizontal differentials.
In more detail, we have that 
\begin{itemize}
	\item The nilpotent differential $\check{\dd}$ is the canonical $G_2$ differential, obtained by projecting the image of the de Rham differential to the relevant $G_2$ irreducible representation;  
     \item $\Psi=e^{-2\phi}\psi$ is the closed 4-form of the integrable $G_2$ structure, and $\cal D$ is the differential defined in equation \eqref{eq:DonQ};
    \item The bottom complex is, up to a sign, the canonical $G_2$ complex of \cite{ReyesCarrionPhd,ReyesCarrion:1998si,fernandez1998dolbeault};
    \item The top complex is isomorphic to the canonical $G_2$ complex (wedge with $\Psi$);
	\item The middle complex is isomorphic (again by wedge with $\Psi$) to the $\cQ$-valued complex defined in \eqref{eq:qcomplex}, with  nilpotent differential $\check \cD$. 
	\item The maps $\Smap^{(p)}$ are the components of the de Rham differential that are projected  out in the canonical $G_2$ differential. $\Tmap^{(-p)}$ are their adjoints, with respect to the shifted symplectic pairing of the complex. These maps are determined by the symmetries and equations of the heterotic $G_2$ system, as detailed below.
\end{itemize}

The main benefit with this presentation of  the total complex, is that it has a very simple shifted symplectic pairing (just wedge and integrate): 
\begin{equation}
\label{eq:shiftsymppair}
    \begin{split}
        	\omega(\xi,\eta)=\langle \xi,\eta\rangle_{\cal E} &=\int \xi\w\eta \;  \; , \;  \;\;  \;\xi,\eta\in \check{\Omega}^{4+*}\oplus\check{\Omega}^*[2]\\
	\omega (X,Y)=\langle X,Y\rangle_{\cal E} &=\int g_\cQ(X\w Y)\;  \; , \; \; X,Y\in\check{\Omega}^*(\cQ)[1]\,.
    \end{split}
\end{equation} 
where $g_\cQ=g\oplus \tr$ is a natural metric on the $\cQ$-bundle, defined in \eqref{eq:innerQ} using the $G_2$ metric, $g$, and taking the trace on vector bundle indices. We note that $\cal D$ is metric with respect to $g_\cQ$.  Moreover, with respect to this pairing, 
 the $\cQ$-row is orthogonal to the top and bottom rows of the complex. 

 The shifted symplectic form induces the signs displayed in the complex. Since $\omega$ is closed under the differential $\newD$ of the complex {(here $|\alpha|$ denotes homological degree)}
 \begin{equation} \label{eq:closedomega}
0 = \newD \omega (\alpha, \beta) =  \omega (\newD\alpha, \beta) + (-1)^{|\alpha|} \omega (\alpha, \newD\beta)  
 \end{equation}
 we find that there is a relative sign between the differentials in the top and bottom row of the complex, and relative signs between the vertical maps.

As a side remark, by defining the Lie algebroid, $\mfg$:
\begin{equation}
	\mfg:= \Omega^0(\cQ)[1]\stackrel{\cD}{\rightarrow}\Omega^1(\cQ)\,,
\end{equation}
we can regard $\check{\Omega}\oplus\check{\Omega}[2]\cong T^*[-1]\check{\Omega}$ and $\check{\Omega}(\cQ)[1]\cong T^*[-1](\mfg)$. From this perspective, the symplectic form defined above is the canonical symplectic form on a cotangent bundle.

\subsection{Defining the differential operator $\newD$ of the complex}
\label{sec:defineD}

We will now return to the precise definition of the differential 
\[
\newD^i : \cE^i \to \cE^{i+1}
\]
of the total complex. To define the correct differential $\newD$ we are guided by the following principles:
\begin{enumerate}
	\item It must be nilpotent;
	\item It must be self-adjoint with respect to $\omega$;
 \item At  homological degree zero, we should recover the action 
 \eqref{eq:action} from $\omega\left( {\cal Z}, \newD {\cal Z}\right) $;
	\item The full complex must be elliptic.
\end{enumerate}
In addition, we can interpret the $\cE^{(-1)}\rightarrow\cE^{(0)}$ as a gauge transformation, so we are here guided by the symmetries of the theory. In particular, the trivial deformations \eqref{eq:trivtrans} should be recovered at homological degree -1.

Our main task now is to determine the differential $\newD^{(i)}$, with the properties just listed. We start at homological degree zero, which should match the results of the last section (condition 3 on the above list). This is in part immediate: we recognize the infinitesimal deformations ${\cal Z}$ of the heterotic $G_2$ system  as elements in 
\[
\cE^0=
\Omega^{4}_{\bf 1}\oplus\Omega^1(\cQ)\oplus\Omega^{2}_{\bf 7} \, .
\]
Also, the horizontal parts of $\newD^{(0)}$  match, up to wedging with $\Psi$, the diagonal terms of self-dual differential operator defined in \eqref{eq:Dcheck}.  {In particular, the signs of the terms in the action, as presented in \eqref{eq:sym2} are consistent with the signs in the top and bottom row of the complex.}  Below, we will show that the maps $\Tmap^{(0)},  \Smap^{(0)}$ match the off-diagonal terms of \eqref{eq:Dcheck} in a similar fashion. 
This ensures that the homological degree zero part of the complex encodes the same equations of motion as the  action \eqref{eq:action}.

At homological degree -1, we are guided by the symmetries  of the theory. These are captured by $\cE^i$ with $i<0$ in the total complex. 
On formal grounds, for nilpotent $\newD$, the action \eqref{eq:action} should have a gauge symmetry corresponding to  $\newD$-exact variations, and this should match the trivial deformations of the heterotic $G_2$ system, \eqref{eq:trivtrans}.
We thus expect:
 \begin{equation} \label{eq:z0triv}
     {\cal Z}^{(0)}_{triv} =   \newD^{(-1)}{\cal Z}^{(-1)} , \mbox{ where } {\cal Z}^{(-1)} \in {\cal E}^{-1} \, = \, \Omega^{0}(Y, {\cal Q}) \oplus \Omega^{1}(Y)~.
 \end{equation}
where the first of the two summands contribute to the trivial  deformations of the $\cal Q$-bundle, and the second summand contain trivial, exact $\beta = \dd \mu$, parametrised by 1-forms $\mu$. These 1-forms  have gauge-for-gauge symmetries, \eqref{eq:trivtrans2} which we expect correspond to
  \begin{equation}
     {\cal Z}^{(-1)}_{triv} = \newD^{(-2)}{\cal Z}^{(-2)} , \mbox{ where } {\cal Z}^{(-2)} \in {\cal E}^{-2} \, = \, \Omega^{0}(Y)~.
 \end{equation}
 These spaces match the first columns of the total complex (or the lower left corner of the bicomplex), but we have yet to construct the differential operators $\newD^{(-1)},\newD^{(-2)}$.

Since the horizontal derivatives have already been determined, we can immediately deduce that $\newD^{(-2)}=-\check{\dd}$. For $\newD^{(-1)}$, we must also determine the maps $\Smap^{(-1)}, \Tmap^{(-1)}$.  Once these maps are determined, we can use self-duality and nilpotency to construct  $\Smap^{(1)}, \Tmap^{(1)}$; this is the topic of the  ensuing subsections. For ease of presentation we will first deduce the three $\Smap$ maps, and then the $\Tmap$, keeping an eye on the symmetries and equations of the heterotic theory along the way, so that these relevant constraints are reproduced by our formalism.

\subsubsection{The $\Smap^{(p)}$ maps}

The map $\Smap^{(-1)}:\Omega^1\rightarrow\Omega^1(\cQ)$ encodes a gauge symmetry of the variations $\beta$.

In our complex, the summand $\check{\Omega}[2]$, in degree zero, corresponds to a linear combination of the $\7$-part of the $B$-field and the $G_2$ structure variations, while  the $\1\4$ part of the $B$-field is embedded inside $\Omega^1(\cQ)$.
The degree -1 entry of $\check{\Omega}[2]$ is the space of gauge transformations of the full $B$-field variation and, consequentially, there is a component of our differential $\Smap^{(-1)}:\Omega^1\rightarrow\Omega^1(TY)\subset \Omega^1(\cQ)$ which, on an arbitrary one form $\mu$, acts by:
\begin{equation} \label{eq:sm1}
	(\Smap^{(-1)}\mu)^a=-(\hat{\dd}\mu)^a \, ,
\end{equation}
where we have defined $\hat{\dd}$ by the general formula $\dd=\check{\dd}+\hat{\dd}$ and the sign is motivated by the sign in the bottom row of the complex. 
In particular, in this case $\hat{\dd}=\pi_{\bf 14}\dd$.

Turning now to $\Smap^{(0)}:\Omega^{2}_{\bf 7}\rightarrow\Omega^6(\cQ)$, we can now be guided by the explicit computation of the second-order action.
On an arbitrary section $\beta\in\Omega^{2}_{\bf 7}$ we have
\begin{equation} \label{eq:sm2}
	(\Smap^{(0)}\beta)^a=-(\hat{\dd}\beta)^a\w\Psi\,,
\end{equation}
where  $\hat{\dd} = \pi_{{\bf 7} \oplus {\bf 27}} \dd$. 
This matches the map $\tau^\dagger$ wedged with $\Psi$, cf. equation  \eqref{eq:T1adj}, as required to reproduce the action \eqref{eq:sym2}.

Finally we have $\Smap^{(1)}:\Omega^{3}_{\bf 1}\rightarrow\Omega^7(\cQ)$. 
Guided by comparison with the maps in lower degree, we posit that
 \begin{equation} \label{eq:sm3}
	(\Smap^{(1)}\chi)^a=-(\hat{\dd}\chi)^a\w\Psi=-({\dd}\chi)^a\w\Psi
\end{equation}
 where, noting that $\check{\dd} \chi=0$, for $\chi \in \Omega^{3}_{\bf 1}$  we have $\hat{\dd}=\dd$. The correctness of this assertion will be verified when we compute its adjoint and see that it is, indeed, encoding a symmetry.
We turn to these computations now.

\subsubsection{The $\Tmap^{(p)}$ maps}
\label{sec:tpmaps}
Given our symplectic form, the maps $\Tmap^{(p)}$ are completely determined by the maps, $\Smap$, which we defined above.
The most important, for us, is that the adjoints $\Tmap^{(-1)}$ and $\Tmap^{(0)}$ are physically sensible. 

We will begin with the map $\Tmap^{(-1)}$ which we demand is adjoint to $\Smap^{(1)}$. First note that the gauge sector of $\cQ$ is irrelevant to this part of the story, and so will be suppressed.
Then, computing explicitly, we have for arbitrary $\sigma\in\Omega^0(TY),\,\chi\in\Omega^{3}_{\bf 1}$: 
\begin{align*}
	\omega( \Tmap^{(-1)}\sigma,\chi )&:=+\omega(  \sigma, \Smap^{(1)}\chi) \; \; \; \; \; \; \;\; \; \;\;\; \; \; \mbox{since } |\sigma|=-1 \mbox{ and } \newD \omega(\alpha,\beta)=0
    \\ 
	&=-\int \sigma^a (\dd\chi)_a\w\Psi \; \; \; \; \; \; \; \;\mbox{sign in $\Smap^{(1)}$ and } \sigma\in\Omega^0(TY)\\
	&=+\int (\sigma^a\Psi)_a\w \dd\chi\; \; \; \; \; \; \; \;\; \mbox{since } \dd \chi \w \Psi = 0\\
	&=+\int \dd((\sigma^a\Psi)_a)\w\chi\\
	&=+\int\pi_{\bf 1}\dd(\sigma^a\Psi_a)\w\chi\,.
\end{align*}
Since this holds for all $\chi$, we conclude that $\Tmap^{(-1)}\sigma=+\pi_{\bf 1}(\dd\sigma^a\Psi_a)\in\Omega^{4}_{\bf 1}$.

Since $\Psi$ is $\dd$-closed, we can re-express this map in terms of the Lie derivative, using Cartan's magic formula, 
\[
\Tmap^{(-1)}\sigma=+\pi_{\bf 1}\cL_\sigma\Psi \, .
\]
The physics interpretation of the space $\Omega^{4}_{\bf 1}$ is that it is the space 
of dilaton variations.
It is manifestly convenient to use the dilaton to redefine the $G_2$ four form, $\psi$ to the closed four form $\Psi$, and thus dilaton variations lead to variations in this form - it is thus more precisely variations of $\Psi$ which are parametrised by $\Omega^{4}_{\bf 1}$.
The group of diffeomorphisms has a natural action on the space of functions (which is where the dilaton lives) and this naturally induces a map out of the Lie algebra of diffeomorphisms to the space of dilaton variations, here encoded by $\Omega^{4}_{\bf 1}$.
This is \emph{precisely} the map $\Tmap^{(-1)}$ that we have identified, so we are satisfied.

Next we turn to $\Tmap^{(0)}$, the adjoint to $\Smap^{(0)}$.
Similar computations to the above yield: 
\begin{align*}
	\omega( \Tmap^{(0)} M,\beta )&:=-\omega( M,\Smap^{(0)}\beta)\; \;\; \;\; \;\; \;\; \;\; \;\; \mbox{since } |M|=0 \mbox{ and } \newD \omega(\alpha,\beta)=0\\
	&=\int M^a\w (\hat{\dd}\beta)_a\w\Psi \; \; \; \;\;\mbox{ recall sign in } \Smap^{(0)}\\
	&=\int (M^a\w \Psi)_a\w \hat{\dd}\beta\\
	&=\int\pi_R(M^a\w\Psi)_a\w \dd\beta\\
	&=-\int \dd\pi_R(M^a\w\Psi)_a\w\beta\,,
\end{align*}
where we defined $\pi_R:\1\oplus\7\oplus\2\7\rightarrow \7\oplus\2\7$. This shows that at homological degree 0, $\omega({\cal Z}, \newD^{(0)} {\cal Z}) $ is equivalent to the action \eqref{eq:action}. Furthermore, through a straightforward, but somewhat tedious, computation one may show that $\Tmap^{(0)}$ agrees with the map $\tau$ in \eqref{eq:T1}; we omit the details here.

Finally, we must compute $\Tmap^{(1)}$:
\begin{align*}
		\omega( \Tmap^{(1)}\eta,\mu) &=\omega(\eta,\Smap^{(-1)}\mu) \\
				&=-\int \eta^a\w (\hat{\dd}\mu)_a  \;\;\;\;\;\; \mbox{ recall sign in } \Smap^{(-1)}\\
	&= +\int \dd\pi_{\bf 14}(\eta^a_{\;\,a})\w \mu\,.
\end{align*}
This map does not have a physical motivation, but this is not an issue, as it is deduced from $\Smap^{(-1)}$, which has been shown encodes a gauge symmetry of $\beta$.

\subsection{Nilpotency and ellipticity}
Having defined the components of a self-adjoint differential, we must now ensure that it is, in fact, nilpotent. 
We can write our differential as a matrix:
\begin{equation}
	\newD=\left(\begin{array}{ccc}
		\check{\dd}&\Tmap & 0\\
		0 &\check{\cal D}&\Smap\\
		0 & 0 &-\check{\dd}
	\end{array}\right)
\end{equation}
{where we redefine our notation so that, in order to match the middle complex of \eqref{diag:defn1},  we have $\check{\cal D} = \{{\cal D}, \Psi \w {\cal D}, {\cal D}\}$ at homological degree $\{-1,0,1\}$. }
This squares to:
\begin{equation}
	\newD^2=\left(\begin{array}{ccc}\check{\dd}^2 \; \;  & \check{\dd}\Tmap+\Tmap\check{\cD} \; &\Tmap\Smap\\
0&\check{\cD}^2&\check{\cD}\Smap-\Smap\check{\dd} \\
0&0&\check{\dd}^2\end{array}\right)\,.
\end{equation}
We know well that the contributions on the diagonal vanish, which means that we need to verify the remaining entries are identically zero.
That means verifying that the compositions $\Tmap^{(p+1)}\circ \Smap^{(p)}$ vanish, 
and that each of the squares in total complex \eqref{diag:defn1} anticommute.

\subsubsection{The $\Tmap\circ \Smap$ compositions}

Nilpotency requires that the compositions $\Tmap^{(1)}\circ \Smap^{(0)}$ and $\Tmap^{(0)}\circ \Smap^{(-1)}$ both vanish.
In fact, it suffices to check only one of these compositions, because they are adjoint to each other.
The argument we give will make this adjointness manifest.

Indeed, non-degeneracy of the symplectic form implies that $\Tmap^{(1)}\Smap^{(0)}\beta=0$ if and only if $\omega (\mu,\Tmap^{(1)}\Smap^{(0)}\beta ) =0$ for all $\mu\in\Omega^1$.
By definition, this is true if and only if $\omega ( \Smap^{(-1)}\mu,\Smap^{(0)}\beta) =0$.
It can be noted that starting with $\Tmap^{(0)}\Smap^{(-1)}$ and taking the adjoint of $\Tmap$ would give precisely the same pairing here.
Expanding this out, we have:
\begin{align*}
	\omega( \Smap^{(-1)}\mu,\Smap^{(0)}\beta) &= \int (\hat{d}\mu)^a\w (\hat{d}\beta)_a\w\Psi\,,
\end{align*}
which is a $G_2$ equivariant map $\1\4\otimes (\7\oplus \2\7)\rightarrow \1$ and must, therefore, be vanishing.

\subsubsection{Anticommuting squares}

The final check of nilpotency is that  each of the squares in diagram \eqref{diag:defn1} anticommute.  
While the computations are straightforward, and similar to the ones presented above (i.e. we will make use of $G_2$ representation theory as above), we include them for completeness. 

\paragraph{Gauge contributions to the squares}

The maps $\Smap$ and $\Tmap$ act trivially on the gauge sector, so it may appear that gauge fields only appear in the $\check{\cD}^2$ computation.
This is not true, however, because $\check{\cD}$ intermixes the gauge and vector field sectors.
Nevertheless, the only possible problems are in checking anticommutativity of the squares; composing the operators we see that one leg trivially vanishes, while the vanishing of the other must be verified.

The trivial directions arise because the $\Tmap$ and $\Smap$ maps are identically zero on gauge fields.
The nontriviality arises in two ways: the $\check{\cD}$ operator maps gauge sector to the vector field sector (relevant for the $\Tmap$ squares); and, the $\check{\cD}$ operator maps vector fields to the gauge sector (relevant for the $\Smap$ squares).

We will start with a section $\alpha\in\Omega^0(\mfg^{ad})\subset \Omega^0(\cQ)$ and must verify that $\Tmap^{(0)}\check{\cD} \alpha =\Tmap^{(0)}\cD\alpha=0$.
Recall from section \ref{sec:Dops} that the vector field component of $\cD\alpha$ is given, up to overall factors, by:
\begin{equation}
	(\cD\alpha)^a\propto \tr(\alpha F^a)\,,
\end{equation}
where $F$ is the gauge curvature, which is a $G_2$ instanton by the equations of motion.

Then:
\begin{align*}
	\Tmap^{(0)}\cD\alpha \propto\pi_{\bf 7}d\pi_R(\tr(\alpha\w F^a)\w\psi'_a) \; 
	=0 \,.
\end{align*}
The result follows trivially from the observation that, for a one form valued in the tangent bundle, say $N$, we have $i_N\psi'$ is independent of the $\1\4$, while $F^a$ is a one form valued in the tangent bundle comprised only of the $\1\4$.
In other words, this map vanishes because there are no nontrivial $G_2$ maps $\1\4\otimes \1\rightarrow \1\oplus\7\oplus\2\7$.

Consider now $\kappa\in\Omega^1(\mfg^{ad})$ and its contribution to $\newD^2$:
\begin{align*}
	\Tmap^{(1)}(\cD\w\Psi)\kappa&\propto d\pi_{\bf 14}(\tr(\kappa\w F^a)\w\Psi)_a\\
	&=d\pi_{\bf 14} (\tr(\kappa_a\w F^a)\w\Psi)=0\,,
\end{align*}
where we used that $F^a\w\Psi_a=0$ to get to the second line, and then used that the resulting expression is clearly in the $\7$, so killed by the projection.

We now turn to the possible contributions in the $\Smap$-squares, where the fact that $\cD$ maps a vector field to a gauge field is the new concern.
In particular:
\begin{align*}
	\check \cD \Smap^{(-1)}\mu&\propto (\hat{d}\mu)^a\w F_a\w\Psi\\
	&=(\hat{d}\mu)^a_{\;\;a}\w F\w\Psi- ((\hat{d}\mu)^a\w F\w\Psi)_a-(\hat{d}\mu)^a\w F\w\Psi_a
	=0\,.
\end{align*}
where we used that contracting an index obeys the Leibniz rule, and that each term vanishes separately due to both $F$ and $\hat{d}\mu$ being in the $\1\4$.

The final gauge theoretic concern is:
\begin{align*}
	\check \cD \Smap^{(0)}\beta =\cD \Smap^{(0)}\beta = (\hat{d}\beta)^a\w F_a\w\Psi \; = 0\,,
\end{align*}
where we have made the simple observation that there are no non-zero $G_2$ maps $(\7\oplus\2\7)\otimes \1\4\rightarrow\1$.

\paragraph{Possible vector field contributions}

The remaining obstructions to nilpotency are vector field contributions, i.e. $\check \cD \Smap -\Smap\check{\dd}$ and $\check{\dd} \Tmap +\Tmap\check \cD $.
Adjointness implies that we need only check the $\Smap$, or respectively $\Tmap$, type squares, but we shall show both as a non-trivial check of our prior computations. 
This means that there are four possible contributions to check.

We will start in the top-left corner of Diagram \ref{diag:defn1}:
\begin{align*}
	\check{\dd}\Tmap^{(-1)}\sigma + \Tmap^{(0)}\cD\sigma  &=\check{\dd}(\pi_{\bf 1}\dd(\sigma^a\Psi_a)) - \check{\dd}(\pi_R (\dd_\zeta\sigma^a\w\Psi)_a)\\
	&=\check{\dd} \left(\pi_{\bf 1} \big( \dd_\zeta\sigma^a\Psi_a+\sigma^a\dd_\zeta\Psi_a\big) -\pi_R ((\dd_\zeta\sigma^a)_a \Psi-\dd_\zeta\sigma^a\w\Psi_a) \right) \\
    &=\check{\dd}\left( (\dd_\zeta\sigma)^a\Psi_a+ \pi_{\bf 1}\sigma^a\dd_\zeta\Psi_a \right)  \\
    &=\check{\dd}\left(  (\dd_\zeta\sigma)^a\Psi_a+ \sigma^a\dd_\zeta\Psi_a\right)\\
    &=\check{\dd} \dd(\sigma^a\Psi_a)= 0\,,
\end{align*}
where in the 2nd to last equality  we use that $\pi_{\bf 1} \dd_\zeta\Psi_a=\dd_\zeta\Psi_a$. The vanishing of this term is the first check that the signs in the complex are correctly defined, given our starting point (minus signs in top row and definition of $\Smap^{(0)}$).

Moving to the next square at the top, the two terms are $\check{\dd}\Tmap^{(0)}$ and $\Tmap^{(1)}\check \cD$.
Computing the first:
\begin{align*}
	\check{\dd}\Tmap^{(0)}M&=-\dd\pi_{\bf 7}\dd\pi_R( M^a\w\Psi)_a=\dd\pi_{\bf 7}\dd\pi_R( M^a\w\Psi_a)\\
	&=\dd\pi_{\bf 7}\dd(M^a\w\Psi_a-\pi_{\bf 1}(M^a\w\Psi_a)) = \dd\pi_{\bf 7}\dd (M^a\w\Psi_a)\, \mbox{, since}\\
	\dd\pi_{\bf 1}(M^a\w\Psi_a)&\propto \dd(\tr M\w\Psi) \propto \dd\tr M\w\Psi  \, .
\end{align*}
and the second:
\begin{align}
	\Tmap^{(1)}\check \cD M = \Tmap^{(1)}\cD M\w\Psi&=+\dd\pi_{\bf 14}\big( (\dd_\zeta M^a\w\Psi)_a\big)\notag\\
	&=\dd\pi_{\bf 14}\left((\dd_\zeta M)^a_{\;\,a}\w\Psi+\dd_\zeta M^a\w\Psi_a\right)\notag\\
	&=\dd\pi_{\bf 14}(\dd_\zeta M^a\w\Psi_a)&&\text{since first term in 7}\,;\notag\\
	\dd(M^a\w\Psi_a)&=\dd_\zeta M^a\w\Psi_a-M^a\w \dd_\zeta\Psi_a\notag\\
	&=\dd_\zeta M^a\w\Psi_a-M^a\w\nabla_a\Psi 
    \notag\\
	&\equiv \dd_\zeta M^a\w\Psi_a\mod \ker \pi_{\bf 14}\label{eq:middle}\\
        \implies \Tmap^{(1)}\cD M\w\Psi&=+\dd \pi_{\bf 14}\dd(M^a\w\Psi_a)\notag
\end{align}
Consequently,
\[
\check{\dd}\Tmap^{(0)} + \Tmap^{(1)}\check \cD = 
+\dd^2 (M^a\w\Psi_a) = 0
\]

The $\Smap$ squares also work out as they should. I.e.
\begin{align*}
	(\check{\cD} S^{(-1)}\mu)_a&
 =-\dd_\zeta((\hat{\dd}\mu))_a\w\Psi
 =-[(\nabla_a\hat{\dd}\mu\Psi)-(\dd\hat{\dd}\mu)_a\w\Psi] =+(\dd\hat{\dd}\mu)_a\w\Psi\,;
\end{align*}
and 
\begin{align*}
	S^{(0)}\check{\dd}\mu&= -(\hat{\dd}\check{\dd}\mu)_a\w\Psi=-(\dd\check{\dd}\mu)_a\w\Psi\,.
\end{align*}
Thus, clearly $\check{\cD} S^{(-1)}\mu-S^{(0)}\check{\dd}\mu=0$.  In the above we have used Lemma 2 of \cite{delaOssa:2017pqy} 
\[
\dd_\zeta z_a = -\dd z_a + \nabla_a z \; , \; z \in \Omega^1(T^*Y) \; .
\]

And finally,  we have the bottom-right square contributions, which come from $S^{(1)}\circ \check{d}$ and $\check{\cD}\circ S^{(0)}$.
The first of these can be computed to be (using the lemma again)
\begin{align*}
	\cD S^{(0)}\beta&= -\dd_\zeta((\hat{\dd}\beta)_a\w\Psi)\\
	&=-\dd_\zeta(\hat{\dd}\beta)_a\w\Psi\\
	&=-\nabla_a\hat{\dd}\beta\w\Psi+(\dd(\hat{\dd}\beta))_a\w\Psi\\
	&=(\dd\hat{\dd}\beta)_a\w\Psi\,;
\end{align*}
and the second:
\begin{align*}
	S^{(1)}\check{\dd}\beta&=-\dd(\check{\dd}\beta)_a\w\Psi\,.
\end{align*}
Again, adding the terms with a relative minus sign gives zero.

We conclude that the  sequence written down in \eqref{diag:defn1}, with $\Smap$ defined as in \eqref{eq:sm1}, \eqref{eq:sm2},\eqref{eq:sm3} and $\Tmap$ defined by duality, has a nilpotent operator $\newD$.\footnote{We also remark that the nilpotency condition for $\newD$, corresponds to a $G_2$ instanton condition for $\mathbb{D}$ defined in \eqref{eq:Dcheck}, though strictly speaking $\mathbb{D}$ is not a connection. Thus the ``instanton for instanton'' property that has been observed for complexes related to heterotic systems \cite{delaOssa:2017gjq,Silva:2024fvl} prevails also for the $\cal{E}$ complex.}

We now have all the information to prove the following proposition
\begin{proposition}
	The sequence \eqref{diag:defn1} is an elliptic complex with shifted symplectic form $\omega$.
\end{proposition}
\begin{proof}
	 see appendix \ref{sec:elliptic}.
\end{proof}

\subsection{The bicomplex}
\label{app:bicomplex}

On formal grounds, the above computations show that we can associate a bicomplex to the theory described by the action \eqref{eq:action}. This was suggested, but not explicitly demonstrated in \cite{Kupka:2024rvl}. The bicomplex is constructed from the following components
\begin{itemize}
    \item  bigraded vector spaces $\cB^{i,j}$
    \item horizontal differentials $\dd_h$
    \item vertical differentials $\dd_v$
    \item the differentials anti-commute $\dd_h\dd_v+\dd_v\dd_h=0$
    \item a shifted symplectic pairing $\omega(\cdot,\cdot)$. 
\end{itemize}

We write the natural bicomplex of the heterotic $G_2$ system as
\[\begin{tikzcd}
	& {-1} & 0 & 1 & 2 \\
	1 & {\Omega^{4}_{\bf 1}} & {\Omega^{5}_{\bf 7}} & {\Omega^6} & {\Omega^7} \\
	0 & {\Omega^0(\cQ)} & {\Omega^1(\cQ)} & {\Omega^6(\cQ)} & {\Omega^7(\cQ)} \\
	{-1} & {\Omega^0} & {\Omega^1} & {\Omega^{2}_{\bf 7}} & {\Omega^{3}_{\bf 1}}
	\arrow[from=2-2, to=2-3]
	\arrow[from=3-2, to=3-3]
	\arrow[from=4-2, to=4-3]
	\arrow[from=4-3, to=4-4]
	\arrow[from=3-3, to=3-4]
	\arrow[from=3-4, to=3-5]
	\arrow[from=2-3, to=2-4]
	\arrow[from=2-4, to=2-5]
	\arrow[from=3-2, to=2-2]
        \arrow[from=3-5, to=2-5]
        \arrow[from=4-2, to=3-2]
	\arrow[from=4-3, to=3-3]
	\arrow[from=3-3, to=2-3]
	\arrow[from=3-4, to=2-4]
	\arrow[from=4-4, to=3-4]
	\arrow[from=4-4, to=4-5]
	\arrow[from=4-5, to=3-5]
\end{tikzcd}\]
where we have defined
\begin{align*}
	\cB^{*,1}&:=\check{\Omega}^{4+*}[1]\,;\\
	\cB^{*,0}&:=\check{\Omega}^*(\cQ)[1] \,;\\
	\cB^{*,-1}&:=\check{\Omega}^*[1]\,,
\end{align*}
where the number in the bracket denotes the shift and the check denotes taking the relevant $G_2$ irreducible representation, for example $\cB^{0,1} = \Omega^{5}_{\bf 7}$. The horizontal arrows are given by the relevant differentials associated to each sub-complex, while the vertical arrows are given by the $\Smap$ and $\Tmap$ maps, except the vertical arrows $\Omega^0\rightarrow\Omega^0({\cal Q})$ and $\Omega^7({\cal Q})\rightarrow\Omega^7$ which are the zero maps.

Associated to the bicomplex is a total complex given by 
 \begin{equation}
     0\rightarrow\cE^{-2}\rightarrow\cE^{-1}\rightarrow
     \cE^{0}\rightarrow\cE^{1}\rightarrow\cE^{2}\rightarrow\cE^{3}\rightarrow 0
 \end{equation}
where 
\[
\cE^i:={\rm Tot}(\cB)^*=\bigoplus_{j+k=i}\cB^{j,k} \, ,
\] 
Of course, this total complex is nothing but the complex \eqref{eq:bicomplex} defined in section \ref{sec:complex}.

\subsection{Infinitesimal moduli and obstruction theory} 

As mentioned at the end of section \ref{sec:dW2}, with the constraint $\beta=0 = \hat{\delta}\phi$, the infinitesimal moduli of the heterotic $G_2$ system are captured by $H^1(Y,\cQ)$  \cite{delaOssa:2017pqy}. The action \eqref{eq:action}, and associated gauge symmetries given by $\newD$-exact forms, 
make it clear that the full infinitesimal moduli are captured by 
\[
    H^0_{\newD}(Y, {\cal E })=\frac{ \{ {\cal Z}^0 \in {\cal E}^0 \; |  \; \newD{\cal Z}^0=0 \} }
    { \{
    {\cal Z}^0 \in {\cal E}^0 \; |  \;  {\cal Z}^0 =\newD{\cal Z}^{-1}
    \} }
    \]
where the superscript denotes homological degree (note that the  fields in ${\cal E}^0$, say, have different form degree).

To understand the higher order deformations requires analysing finite order deformations $\cal Z$ of the superpotential of a background solution. The corresponding equation of motion is then a Maurer-Cartan equation of the form
\begin{equation}
\label{eq:MCeq}
    \check{\mathbb D}{\cal Z}+{\mathcal O}({\cal Z}^2)=0\:.
\end{equation}
Solutions to this equation parameterises finite deformations of the system, i.e. which of the infinitesimal moduli can be integrated to a finite deformation. Figuring out the explicit form of the higher order terms ${\mathcal O}({\cal Z}^2)$ requires an investigation of the full deformation theory of heterotic $G_2$ systems and its associated $L_\infty$-algebra. This is beyond the scope of this paper. We can however make some remarks. 

Note first that a perturbative expansion of the Maurer-Cartan equation \eqref{eq:MCeq} together with general deformation theory arguments leads to an obstruction, or Kuranishi, map 
\begin{equation*}
    {\cal O}\::\;H^0_{\newD}(Y, {\cal E })\rightarrow H^1_{\newD}(Y, {\cal E })\:,
\end{equation*}
where the un-obstructed deformations lie in the kernel of this map. The details of this map depends of course on the deformation problem at hand. For ``generic'' Maurer-Cartan equations, the expectation is that this map will be injective for $h^0_{\newD}(Y, {\cal E })\le h^1_{\newD}(Y, {\cal E })$ and surjective for $h^0_{\newD}(Y, {\cal E })\ge h^1_{\newD}(Y, {\cal E })$. In particular, for $h^0_{\newD}(Y, {\cal E }) = h^1_{\newD}(Y, {\cal E })$ the map is ``generically'' a bijection, just like a generic quadratic matrix is invertible. The expected, or ``virtual'', dimension of the moduli space in this case is zero, consisting of (possibly infinitely many) points. If this holds, and the moduli space can be compactified in some appropriate sense, there is then a natural enumerate invariant for topological spaces admitting heterotic $G_2$ solutions. Namely, how many such solutions exist for a given topology?

We thus want to establish the following:
\begin{proposition}
    For the heterotic $G_2$ system, $h^0_{\newD}(Y, {\cal E })= h^1_{\newD}(Y, {\cal E })$. 
\end{proposition}
\begin{proof}
That this is true follows from the existence of the non-degenerate pairing $\omega$, and Hodge theory. To invoke Hodge theory, we need to introduce a positive definite pairing $\tilde g$ on the complex. Using the same notation as in \eqref{eq:shiftsymppair} for the elements in $\cal{E}$, we pick the following pairing
\begin{equation}
\label{eq:PosDefPairing}
    {\tilde g}({\cal Z}_1,{\cal Z}_2)=\int \tilde g_{\cal Q}(X_1\wedge*X_2)+\xi_1\wedge*\xi_2+\eta_1\wedge*\eta_2=\int \tilde g_{AB}{\cal Z}^A_1\wedge*{\cal Z}^B_2\:,
\end{equation}
where $\tilde g_{\cal Q}$ is the pairing on $\cal Q$ defined in \eqref{eq:innerQ}, but with the sign of the trace flipped on the $\End(TY)$ part, so that it is positive definite. The reason for this choice will become clear below.

Let us first compute how the adjoint of ${\newD}$ acts with respect to this inner product. This is a bit subtle, as we need to invoke $\omega$ to perform the integration by parts. Let $\omega_{AB}$ denote the point-wise components of the  pairing $\omega$, with inverse $\omega^{AB}$. We compute
\begin{align}
    \tilde{g}(\newD{\cal Z}_1,{\cal Z}_2)&=\int\newD{\cal Z}^A_1\wedge*{\cal Z}_2^B\tilde g_{AB}\notag\\
    &=\int\newD{\cal Z}^A_1\wedge*{\cal Z}_2^B\tilde g_{BC}\omega^{CD}\omega_{DA}\notag\\
    &=(-1)^{\vert{\cal Z}_1\vert}\int{\cal Z}^A_1\wedge\newD\left(*{\cal Z}_2^B\tilde g_{BC}\omega^{CD}\right)\omega_{DA}\:.
\end{align}
It is then clear that for ${\cal Z}$ in the kernel of the adjoint requires
\begin{equation}
    \newD*\left(\omega^{-1}\circ\tilde g\cdot{\cal Z}\right)=0\:.    
\end{equation}

Let us define the map 
\begin{equation}
    \Theta\::\;\;\;{\cal E }^0\rightarrow{\cal E }^1\:,\;\;\;\;\Theta({\cal Z})=*\left(\omega^{-1}\circ\tilde g\cdot{\cal Z}\right)\:,
\end{equation}
which is an isomorphism. If we now take ${\cal Z}$ to be harmonic, then  $\Theta({\cal Z})$ is $\newD$-closed. To check coclosure, we must compute the adjoint action on $\Theta({\cal Z})$. We get
\begin{equation}
    \newD*\left(\omega^{-1}\circ\tilde g\cdot\Theta({\cal Z})\right)=\newD\left(\omega^{-1}\circ\tilde g\circ\omega^{-1}\circ\tilde g\cdot{\cal Z})\right)\:.
\end{equation}
For our particular choice of $\tilde g$, it is straight forward to check that $\omega^{-1}\circ\tilde g\circ\omega^{-1}\circ\tilde g$ is the identity. Hence
\begin{equation}
    \newD*\left(\omega^{-1}\circ\tilde g\cdot\Theta({\cal Z})\right)=\newD{\cal Z}=0\:, 
\end{equation}
so $\Theta({\cal Z})$ is both closed and co-closed.

It follows that the map $\Theta$ maps harmonic forms to harmonic forms. This establishes the above isomorphism between  $H^0_{\newD}(Y, {\cal E })$ and $H^1_{\newD}(Y, {\cal E })$. 
\end{proof}

As a corollary, by a similar argument, one can also establish the more general Serre duality-like result
\begin{equation}
h^{p}_{\newD}(Y, {\cal E })= h^{1-p}_{\newD}(Y, {\cal E })\:.    
\end{equation}
The corresponding Euler characteristic, or ``index''
\begin{equation}
    \chi(\newD,{\cal E})=\sum_{p=-2}^3h^{p}_{\newD}(Y, {\cal E })\:,
\end{equation}
therefore vanishes. This shows that the moduli problem of the heterotic $G_2$ system is an ``index zero moduli problem''. This is encouraging, as such moduli problems have a tendency to lead to mathematically well-behaved topological invariant theories and enumerative geometry, such as, for example, in the corresponding moduli problems of ordinary Chern-Simons theory \cite{Witten:1988hf}, holomorphic Chern-Simons theory \cite{donaldson1998gauge, thomas1997gauge}, and more speculatively $G_2$ instantons \cite{donaldson2009gauge}.

\subsection{Field redefinitions and diagonalization of $\newD$}
\label{sec:diagonal}

While the ${\cE^*}$ complex, and the associated bicomplex, provide a good description of the deformation theory, the fact that $\newD$ has off-diagonal components will complicate our analysis in section \ref{sec:1loopall}. Thus, it is convenient to look for a field-redefinition which decouples the rows of the ${\cE^*}$ complex. In this section we show how to accomplish this. To do so, we first enlarge the ${\cE^*}$ complex to a complex which does not include projections in the top and bottom rows. We call this complex ${\cal F}^*$. It takes the form
\begin{equation}\label{diag:defn2}
	\begin{tikzcd}
		\dots\arrow{r}{\dd} & \Omega^3\arrow{r}\arrow{r}{\dd} &\Omega^4\arrow{r}{\dd}&\Omega^5\arrow{r}{\dd}&\Omega^6\arrow{r}{\dd}&\Omega^7\\
		& \Omega^0(\cQ)\arrow{r}{\cD}\arrow{ru}{\Tmap^{(-1)}}&\Omega^1(\cQ)\arrow{r}{\Psi\w\cD}\arrow{ru}{\Tmap^{(0)}}&\Omega^6(\cQ)\arrow{r}{\cD}\arrow{ru}{\Tmap^{(1)}}&\Omega^7(\cQ)\\
		\Omega^0\arrow{r}{-\dd}&\Omega^1\arrow{r}{-\dd}\arrow{ru}{\Smap^{(-1)}}&\Omega^2\arrow{r}{-\dd}\arrow{ru}{\Smap^{(0)}}&\Omega^3\arrow{r}{-\dd}\arrow{ru}{\Smap^{(1)}} &\Omega^4\arrow{r}{-\dd} & \dots
	\end{tikzcd}
\end{equation}
For this to still give a bicomplex, we need to be careful with how we define the $\Smap$  and $\Tmap$ maps. The $\Tmap$ maps are now defined with additional projectors onto the correct irreducible representations in the top complex, while the $\Smap$ maps only act on the appropriate irreducible representations on the bottom complex. Specifically,
\begin{align}
    {\Tmap}_{\rm new}&=\pi\circ{\Tmap}_{\rm old}\\
    {\Smap}_{\rm new}&={\Smap}_{\rm old}\circ\pi\:,
\end{align}
where $\pi$ denotes the projection onto the appropriate representation in the old complex $\cE^*$. 

Using the new complex ${\cal F}^*$, we would now like to find a field redefinition which diagonalizes the complex. We will focus on the degree zero fields, as these are the classical fields in the action \eqref{eq:action}. The operator of the complex ${\cal F}^*$ now reads
\begin{equation}
	\newD=\left(\begin{array}{ccc}
		\dd&\Tmap & 0\\
		0 &\check{\cal D}&\Smap\\
		0 & 0 &-\dd
	\end{array}\right)=
        \left(\begin{array}{ccc}
		\mathbf{d}&\Tmap & 0\\
		0 & \mathbf{d} &\Smap\\
		0 & 0 & \mathbf{d}
	\end{array}\right)\:,
\end{equation}
where we have denoted the differential on each row as $\mathbf{d}$ for a smoother computation (we can reinsert the proper differentials at any stage in the computation below). The pairing on ${\cal E}$ extends naturally to ${\cal F}$, so we can write the action as
\begin{equation}
     S =  \langle {\cal Z}, \newD\,{\cal  Z} \rangle_{\cal F} \; ,
\end{equation} 
for $\cZ$ in the image of the natural inclusion from $\cE$.

We would like to find maps $G$ and $G^{-1}$ so that the operator diagonalizes as 
\begin{equation}
    G^{-1}\circ\newD\circ G=
    \left(\begin{array}{ccc}
		\mathbf{d} & 0 & 0\\
		0 & \mathbf{d} & 0 \\
		0 & 0 & \mathbf{d}
	\end{array}\right)\:.
\end{equation}
The map $G$ then corresponds to a field redefinition of ${\cal Z}$, which will follow from what we show below. 
Let us postulate the following form of the map $G$ 
\begin{equation}
    G=
    \left(\begin{array}{ccc}
		1 & B & C\\
		0 & 1 & A \\
		0 & 0 & 1
	\end{array}\right)\:.
\end{equation}
Then the inverse can easily be checked to be
\begin{equation}
    G^{-1}=
    \left(\begin{array}{ccc}
		1 & - B & BA-C\\
		0 & 1 & - A \\
		0 & 0 & 1
	\end{array}\right)\:.
\end{equation}
Let us then compute
\begin{equation}
    G^{-1}\circ\newD\circ G=
    \left(\begin{array}{ccc}
		\mathbf{d} \;&\; \Tmap + [\mathbf{d}, B] \;\;&\;\; [\mathbf{d}, C] + \Tmap\circ A + B\circ\left([A,\mathbf{d}]-\Smap \right)\\
		0  & \mathbf{d} & \Smap + [\mathbf{d},A] \\
		0  & 0 & \mathbf{d}
	\end{array}\right)\:.
\end{equation}
We will show below that the map $A$ can be chosen so that 
\begin{equation}
\label{eq:TrivS}
    [\mathbf{d},A] + \Smap =0\:.
\end{equation}
Furthermore, we can also choose the map $B$ so that
\begin{equation}
\label{eq:TrivT}
    [\mathbf{d}, B]+ \Tmap=0\:.
\end{equation}
With these choices for $A$ and $B$, the expression then simplifies to
\begin{equation}
\label{eq:AlmostTrivD}
    G^{-1}\circ\newD\circ G={\small
    \left(\begin{array}{ccc}
		\mathbf{d} \;&\; 0 \;\;&\;\; [\mathbf{d}, C] - [\mathbf{d}, B]\circ A \\
		0  & \mathbf{d} & 0 \\
		0  & 0 & \mathbf{d}
	\end{array}\right)=
    \left(\begin{array}{ccc}
		\mathbf{d} \;&\; 0 \;\;&\;\; [\mathbf{d}, C-BA] - B\circ \Smap \\
		0  & \mathbf{d} & 0 \\
		0  & 0 & \mathbf{d}
	\end{array}\right)}\:,
\end{equation}
where we have used \eqref{eq:TrivS}.

Let us now see that we can trivialise the $\Tmap^{(0)}$ map as in \eqref{eq:TrivT}. Recall that for $Z=(M,\alpha)\in\Omega^1({\cal Q})$ we have
\begin{align}
    \Tmap^{(0)}Z&=-\pi_{\bf 7}\circ\dd\left(\pi_{\7+\2\7}M^a\wedge\Psi_a\right)\notag\\
    &=-\dd\left(\pi_{\7+\2\7}(M^a\wedge\Psi_a)\right)+\pi_{\bf 14}\dd\left(\pi_{\7+\2\7}M^a\wedge\Psi_a\right)\notag\\
    &=-\dd\left(\pi_{\7+\2\7}(M^a\wedge\Psi_a)\right)+\pi_{\bf 14}\dd\left(M^a\wedge\Psi_a\right)\notag\\
    &=-\dd\left(\pi_{\7+\2\7}(M^a\wedge\Psi_a)\right)+\pi_{\bf 14}\left((\dd_{\zeta} M^a)\wedge\Psi_a\right)\notag\\
    &=-\dd\left(\pi_{\7+\2\7}(M^a\wedge\Psi_a)\right)+\pi_{\bf 14}\left((\dd_{\zeta} M^a\wedge\Psi)_a\right)\notag\\
    &=-\dd\left(\pi_{\7+\2\7}(M^a\wedge\Psi_a)\right)+\pi_{\bf 14}\left((({\cal D}Z)^a\wedge\Psi)_a\right)\notag\\
    &=-\dd\left(\pi_{\7+\2\7}(Z^a\wedge\Psi_a)\right)+\pi_{\bf 14}\left(((\check{\cal D}Z)^a)_a\right)\:,
\end{align}
where in the third equality we have used that the singlet drops out from the second term, the fourth equality follows the computation in \eqref{eq:middle}, the $\7$-part again drops out in the fifth equality, while the gauge contribution drops out in the sixth equality. This then defines the maps $B^{(0)}$ and $B^{(1)}$ so that we may write
\begin{equation}
    \Tmap^{(0)}=-[\mathbf{d},B]^{(0)}=-\mathbf{d}\circ B^{(0)} + B^{(1)}\circ\mathbf{d}\:.
\end{equation}
Note that the projectors in $B^{(0)}$ and $B^{(1)}$ go outside of the original complex ${\cal E}^*$. This is why we needed to expand our complex to ${\cal F}^*$ in order to diagonalise the differential. The $\Smap^{(0)}$ can similarly be trivialised as in \eqref{eq:TrivS}. Indeed, this follows readily as $\Smap^{(0)}$ is again the adjoint of $\Tmap^{(0)}$ with respect to the pairing on ${\cal F}^*$.

Consider next the upper right-hand corner of the last matrix in \eqref{eq:AlmostTrivD}. We have for $\beta\in\Omega^{2}$
\begin{equation}
    B^{(1)}\circ \Smap^{(0)}\beta=B^{(1)}\left(-(\hat\dd\pi_{\bf 7}(\beta))\right)=0\:.
\end{equation}
This follows as $\hat\dd$ projects onto the $\7+\2\7$ irreducible representation when acting on $2$-forms, while $B^{(1)}$ projects onto the $\1\4$ irreducible representation. It follows that
\begin{equation}
        G^{-1}\circ\newD\circ G=
    \left(\begin{array}{ccc}
		\mathbf{d} \;&\; 0 \;\;&\;\; [\mathbf{d}, C - B\circ A] \\
		0  & \mathbf{d} & 0 \\
		0  & 0 & \mathbf{d}
	\end{array}\right)\:.
\end{equation}
We may then simply pick $C=B\circ A$ to get a diagonalised differential.

We have succeeded in diagonalising the differential $\newD$, at least in degree zero, but at a cost. We needed to expand our complex to ${\cal F}^*$. Furthermore, the field redefinitions we do, given by $G$ and $G^{-1}$, add parts of irreducible representations to our fields which are not in the original complex ${\cal E}^*$. Hence, we are not quite done yet; we need to ensure that the ${\cal F}^*$ complex does not introduce unphysical degrees of freedom. We will resolve this question in section \ref{sec:1loop}, by showing that the classical theory described in terms of the diagonal operator is equivalent to the original theory, and that the result also holds for the one-loop partition function.

\section{1-loop partition function}
\label{sec:1loopall}

We have seen that the second order variations heterotic $G_2$ superpotential can be interpreted as the action \eqref{eq:action},  which we repeat here for convenience, 
\begin{equation}
    \label{eq:1loopAction}
 S=  \langle {\cal Z}, \newD\,{\cal  Z} \rangle_{\cal E} ~.
\end{equation} 
We will now use this quadratic action to compute the absolute value of the one-loop partition function of the heterotic $G_2$ system. For a quadratic action, this can be done using formal path integral manipulations, resulting in a rigorous end result. Our analysis will be modelled on the corresponding computation for three-dimensional Chern-Simons theory \cite{Witten:1988hf}, which we repeat for the reader's convenience. The reader is referred to \cite{Pestun:2005rp} for a more extended review.

\subsection{Warm up: Abelian Chern-Simons theory}
To set the stage, we start by recalling Abelian Chern-Simons theory on a three-manifold ${\cal M}_3$ \cite{Witten:1988hf}
\begin{equation}
    S[\alpha]=\int_{{\cal M}_3}\alpha\wedge\mathrm{d}\alpha\:.
\end{equation}
The absolute value of the partition function is formally defined as the Euclidean path integral
\begin{equation}
    \vert Z\vert=\frac{1}{{\rm Vol}(G)}\int{\cal D}\alpha\: e^{-S[\alpha]}\:,
\end{equation}
where ${\rm Vol}(G)$ denotes the infinite volume of gauge transformations to be determined. 

Picking a metric, we can filter the path-integral under the Hodge-decomposition as 
\begin{align}
    \vert Z\vert&=\frac{1}{{\rm Vol}(G)}\int_{\mathrm{d}\Omega^0}{\cal D}\alpha\int_{\mathrm{d}^\dagger\Omega^2}{\cal D}\alpha\: e^{-S[\alpha]}\notag
    =\frac{{\rm Vol}(\mathrm{d}\Omega^0)}{{\rm Vol}(G)}\int_{\mathrm{d}^\dagger\Omega^2}{\cal D}\alpha\: e^{-S[\alpha]}\notag\\[5pt]
    &=\frac{{\rm Vol}(\mathrm{d}\Omega^0)}{{\rm Vol}(G)}\frac{1}{\sqrt{{\det\left(\mathrm{d}\vert_{\mathrm{d}^\dagger\Omega^2}\right)}}}\:,
\end{align}
where we use that $S[\alpha]$ is independent of exact modes, and we ignore potential finite-dimensional contributions from harmonic on-shell modes. The determinant is defined as
\begin{equation}
\det\left(\mathrm{d}\vert_{\mathrm{d}^\dagger\Omega^2}\right):=\det\left(\mathrm{d}^\dagger\mathrm{d}\vert_{\mathrm{d}^\dagger\Omega^2}\right)^{\tfrac12}\:.
\end{equation}
Note then that, due to standard Hodge theory, the determinant of the Laplacian on one-forms, $\Delta^1$, decomposes as
\begin{equation}
\det\left(\Delta^1\right)=\det\left(\mathrm{d}\mathrm{d}^\dagger\vert_{\mathrm{d}\Omega^0}\right)\det\left(\mathrm{d}^\dagger\mathrm{d}\vert_{\mathrm{d}^\dagger\Omega^2}\right)\:.
\end{equation}
The partition function can therefore be written as
\begin{equation}
    \vert Z\vert=\frac{{\rm Vol}(\mathrm{d}\Omega^0)}{{\rm Vol}(G)}\frac{\det\left(\mathrm{d}\mathrm{d}^\dagger\vert_{\mathrm{d}\Omega^0}\right)^{\tfrac{1}{4}}}{{\det\left(\Delta^1\right)}^{\tfrac14}}\:.
\end{equation}
Let us then consider $\det\left(\mathrm{d}\mathrm{d}^\dagger\vert_{\mathrm{d}\Omega^0}\right)$. If $\alpha=\mathrm{d}\gamma$ is an eigenvector of $\mathrm{d}\mathrm{d}^\dagger\vert_{\mathrm{d}\Omega^0}$, then
\begin{equation}    \mathrm{d}\mathrm{d}^\dagger\mathrm{d}\gamma=\lambda\mathrm{d}\gamma\:.
\end{equation}
We can assume that $\gamma\in\mathrm{Im}(\mathrm{d}^\dagger)$, and $\mathrm{d}$ is invertible on this image, so this equation is equivalent to
\begin{equation}
\mathrm{d}^\dagger\mathrm{d}\gamma=\Delta^0\gamma=\lambda\gamma\:.
\end{equation}
It follows that $\det\left(\mathrm{d}\mathrm{d}^\dagger\vert_{\mathrm{d}\Omega^0}\right)=\det(\Delta^0)$.

Next let's consider ${\rm Vol}(\mathrm{d}\Omega^0)$. Elementary linear algebra states that
\begin{equation}
    {\rm Vol}(\mathrm{d}\Omega^0)=\frac{\det\left(\mathrm{d}\vert_{\mathrm{d}^\dagger\Omega^1}\right)}{\ker(\dd:\Omega^0\rightarrow\Omega^1)}{\rm Vol}(\mathrm{d}\Omega^0)=\det\left(\Delta^0\right)^{\tfrac12}{\rm Vol}(\Omega^0)\:,
\end{equation}
where we have ignored the finite-dimensional kernel of constant functions. Collecting everything, one finds
\begin{equation}
    \vert Z\vert=\frac{{\rm Vol}(\Omega^0)}{{\rm Vol}(G)}\frac{\det\left(\Delta^0\right)^{\tfrac34}}{\det\left(\Delta^1\right)^{\tfrac14}}\:.
\end{equation}
By choosing ${\rm Vol}(G)={\rm Vol}(\Omega^0)$, as the volume of gauge transformations, we get
\begin{equation} \label{eq:CS1loop}
    \vert Z\vert=\frac{\det\left(\Delta^0\right)^{\tfrac34}}{\det\left(\Delta^1\right)^{\tfrac14}}\:.
\end{equation}
Upon regularisation of the determinants of the Laplacians, this is precisely the Ray-Singer torsion of the manifold ${\cal M}_3$.

\subsection{1-loop partition function of the heterotic $G_2$ system}
\label{sec:1loop}
The computation of the absolute value of the partition function of the theory \eqref{eq:1loopAction} proceeds in exactly the same fashion as for three-dimensional Chern-Simons theory. The only difference being that the theory \eqref{eq:1loopAction} is one-reducible, i.e. it has gauge of gauge symmetries, in addition to gauge symmetries. Keeping this in mind, and the fact that our fields take values in level zero $\cE^0$ of the complex \eqref{eq:bicomplex}, the formal computation of the (absolute value of the) partition function gives
\begin{equation}
    \vert Z\vert =\frac{\det\left(\Delta_{\newD}^{(-1)}\right)^{\tfrac34}}{\det\left(\Delta_{\newD}^{(-2)}\right)^{\tfrac54}\det\left(\Delta_{\newD}^{(0)}\right)^{\tfrac14}}\:.
\end{equation}
This is the same expression as the partition function computed in \cite{Kupka:2024rvl} using generalised geometry techniques, except that the differential $\newD$ is slightly different from the one appearing in \cite{Kupka:2024rvl}. We expect the results to agree upon an appropriate field redefinition of the quadratic action.

Here the Laplacians are defined as
\begin{equation}
    \Delta_{\newD}=\newD\newD^\dagger+\newD^\dagger\newD\:,
\end{equation}
where the adjoint $\newD^\dagger$ is defined using a {\it positive definite} inner product on $\cE^*$. For example, we can use the inner product given by \eqref{eq:PosDefPairing} 
\begin{equation}
    ({\cal Z}_1,{\cal Z}_2)=\int \big( \tilde g_{\cal Q}(X_1\wedge*X_2)+\xi_1\wedge*\xi_2+\eta_1\wedge*\eta_2 \big)\:,
\end{equation}
where ${\cal Z}_i=(\xi_i,X_i,\eta_i)\in\cB^{*,1}\oplus\cB^{*,0}\oplus\cB^{*,-1}$. Note that this is not the pairing $\omega$, which is not positive definite, and therefore does not give rise to an appropriate Hodge decomposition.  

We would like to be a bit more explicit, and  write this partition function in terms of determinants of more familiar operators. To do so, it is convenient to use the field-redefinition introduced in section \ref{sec:diagonal}, which decouples the rows of the ${\cE^*}$ complex.
To proceed, recall that the original action \eqref{eq:1loopAction} is invariant under shifts by $\newD$-exact forms of the fields. This is in fact the gauge symmetry of the heterotic theory. In particular, fields in the lowest row of the ${\cal E}^*$ complex are shifted by $\check{\dd}$-exact forms. Let $\beta\in\Omega^{2}_{\bf 7}$ correspond to this part of the of the field ${\cal Z}$. If we ignore a finite-dimensional part given by harmonic forms, we can Hodge decompose $\beta$ as
\begin{equation}
\label{eq:HodgeBeta1}
    \beta=\check{\dd}^{\tilde{\dagger}}(f\Phi)+\check{\dd}\gamma\:,
\end{equation}
where $\gamma_1\in\Omega^1$, and $\Phi=\tilde{*}\Psi$ is the Hodge-dual of $\Psi$ with respect to the metric given by $\Psi$. The adjoint is also defined with respect to this metric. Without loss of generality, we can pick a gauge where $\beta$ does not have any $\check{\dd}$-exact part (this does not change the action). Since the field redefinitions corresponding to $G$ and $G^{-1}$ do not touch the bottom row of ${\cal E}^*$, it leaves $\beta$ unchanged. It follows that the gauge choice for $\beta$ is preserved: the $\check{\dd}\gamma$-term cannot appear in the action after the field redefinition. 

Let $\alpha\in\Omega^{4}_{\bf 1}$ be the top-row part of ${\cal Z}$, and $\tilde\alpha\in\Omega^4$ be the part of the field after the field redefinition. This $\tilde\alpha$ need not be a singlet. The part of the action involving $\beta$ after a field redefinition is then 
\begin{equation}
    S[\beta,\tilde\alpha]=\int\beta\wedge\dd\tilde\alpha\:.
\end{equation}
Since the $\check{\dd}$-exact part of $\beta$ does not appear, this term is the same as 
\begin{align}
    S[\beta,\tilde\alpha]&=\int\check{\dd}^{\tilde{\dagger}}(f\Phi)\wedge\dd\tilde\alpha=({\dd}^{\tilde{\dagger}}(f\Phi),\tilde*\dd\tilde\alpha)=(\tilde*\dd(f\Psi),\tilde*\dd\tilde\alpha)\notag\\
    &=(\dd(f\Psi),\dd\tilde\alpha)=(f\Psi,\dd^{\tilde\dagger}\dd\tilde\alpha)=(f,\Psi\lrcorner\dd^{\tilde\dagger}\dd\tilde\alpha)\:,
\end{align}
where $(\:,\:)$ is the positive definite pairing induced by the metric $\tilde g$ of $\Psi$, and we have not been careful with signs. We can assume without loss of generality that $f$ does not have a constant part (as it would drop out of \eqref{eq:HodgeBeta1}). If we Hodge decompose $\Psi\lrcorner\dd^{\tilde\dagger}\dd\tilde\alpha$, we get
\begin{equation}
    \Psi\lrcorner\dd^{\tilde\dagger}\dd\tilde\alpha=\dd^{\tilde\dagger}\dd g+{\rm constant}\:.
\end{equation}
Note that while $g$ is a singlet, it can get contributions from the $(\7+\2\7)$-representations of the middle row of the complex, via the field redefinition. 

We thus get 
\begin{align}
    S[\beta,\tilde\alpha]&=(f,\dd^{\tilde\dagger}\dd g)=(\dd f,\dd g)=(\Psi\wedge\dd f,\Psi\wedge\dd g)\notag\\
    &=(\tilde*\dd(f\Psi),\tilde*\dd(g\Psi))=\int\beta\wedge\dd(g\Psi)\:.
\end{align}
This shows that the $\7$-valued two-form $\beta$ can only couple to a singlet, also after the field redefinition! As the singlet has no gauge transformations, except for a finite-dimensional space of translations by constants which we ignore, we will take this as our redefined singlet in the top row of the original complex ${\cal E}^*$, now using the diagonal differential on ${\cal E}^*$.
By slight abuse of notation we shall also refer to this singlet as $\alpha\in\Omega^{4}_{\bf 1}$. 

The above computations shows that the original classical action \eqref{eq:1loopAction} can be written as
\begin{equation}
\label{eq:1loopAction2}
    S=\langle Z,\check{\cal D}Z\rangle_{\cal Q}+\int\beta\wedge\check{\dd}\alpha\:,
\end{equation}
for $\alpha\in\Omega^{4}_{\bf 1}$. The gauge transformations in the current form of the action the $\beta$-fields are precisely as before, shifts by $\check{\dd}$-exact forms. The gauge symmetries of \eqref{eq:1loopAction2} are given by $\mathbf{d}$-exact forms, where $\mathbf{d}$ refers to the diagonal differential introduced in section \ref{sec:diagonal}.

We had to pick a particular $\newD$-gauge to write the action in the form of \eqref{eq:1loopAction2}. However, the field redefinitions we have done only include translations corresponding to the upper-triangular matrices $G$ and $G^{-1}$, which therefore have a trivial Jacobian in the change of the path integral measure. A general $\newD$-gauge transformation will add terms to the action \eqref{eq:1loopAction2}, but these terms can, as the above computations show, be absorbed again by field redefinitions, which again, do not change the path integral measure. The theories are therefore equivalent as quantum theories.

The action \eqref{eq:1loopAction2} is  simpler to quantise in terms of determinants of more familiar Laplacians. Furthermore, it is easier to study its BV quantisation, $\eta$-invariants, etc. We will leave this for future work, but compute the (absolute value of the) partition function here. A straight-forward computation then gives 
\begin{equation}
\label{eq:Simplified1loop}
    \vert Z\vert = \frac{\left(\det\Delta_{\check{\cal D}}^{(-1)}\right)^{\tfrac34}}{\det\left(\Delta_{\check{\cal D}}^{(0)}\right)^{\tfrac14}}\times\frac{\left(\det\Delta_{\bf 7}\right)^{\tfrac12}}{\left(\det\Delta_{\bf 1}\right)^{\tfrac32}}\:,
\end{equation}
where the superscript denotes the Laplacian at the corresponding entry in the ${\cal E}^*$ complex, and the subscripts denote the irreducible representation. 
Here we have used that the Laplacians acting on one-forms and $\7$-valued two-forms are equivalent, and similarly with functions and singlet valued four-forms \cite{Ashmore:2021pdm}.

It is now straight-forward to compare the factors in \eqref{eq:Simplified1loop} with the terms in the action \eqref{eq:1loopAction2}. The first factor in \eqref{eq:Simplified1loop} derives from the first term in \eqref{eq:1loopAction2}, and is the one-loop partition function of an instanton gauge theory \cite{deBoer:2006bp}, with instanton connection $\cal D$ on the bundle $\cal Q$. The second factor derives from the second term in \eqref{eq:1loopAction2}. It is the inverse of the square of the result of an Abelian instanton gauge theory. The square comes from the fact that we have two fields, $\beta$ and $\alpha$, while the inverse is due to the fields being shifted by one form degree. The result may also be compared with previous quantization of type II strings on $G_2$ manifolds \cite{deBoer:2005pt, deBoer:2007zu, Ashmore:2021pdm}.

\section{Conclusion and outlook}
\label{sec:concl}

In this paper, we have used the superpotential we defined in \cite{delaOssa:2019cci}, for three-dimensional, ${\cal N}=1$ compactifications of heterotic supergravity on 7-manifolds $Y$ with $G_2$ structure, to develop a deformation theory of such heterotic $G_2$ systems. We have shown that the vanishing of the second order variation of the superpotential encodes the equations for infinitesimal moduli of these string configurations, as the constraint 
\[ \newD {\cal Z} = 0 \; .\] 
Here ${\cal Z} \in \cE^0=
\Omega^{4}_{\bf 1}\oplus\Omega^1(\cQ)\oplus\Omega^{2}_{\bf 7}$ is constructed from the deformations of the supergravity fields (dilaton, metric, gauge field and $B$-field). We have shown, by explicitly constructing the relevant nilpotent differential $\newD$, and the required symplectic pairing, that this constraint can be  encoded in an elliptic double complex (or the associated total complex). This is in accordance with more formal reasoning, based on generalised geometry, for heterotic $G_2$ compactifications  \cite{Kupka:2024rvl}.  We have also noted that the nilpotency of $\newD$ can be interpreted as a $G_2$ instanton condition, as has been observed previously for complexes related to heterotic systems \cite{delaOssa:2017gjq,Silva:2024fvl}.

Our bicomplex encodes the equations for infinitesimal moduli of heterotic $G_2$ systems at zero homological degree. At lower homological degree one finds the gauge, and gauge for gauge, symmetries of the theory. It thus follows that the infinitesimal moduli of heterotic $G_2$ systems are captured by $H^0_\newD(Y,{\cal E})$ where the superscript encodes homological degree. Moreover, the obstruction space is given by $H^1_\newD(Y,{\cal E})$. We have shown that the dimensions of these two cohomology groups are the same, showing that the expected, or ``virtual'', dimension of the infinitesimal moduli space is zero.  This generalizes to a Serre-like duality and shows that the complex has vanishing Euler characteristic $\chi(\newD,{\cal E})$. This indicates that there might be natural enumerative invariants for heterotic $G_2$ systems. Computing $H^p_\newD(Y,{\cal E})$ in explicit solutions, such as \cite{delaOssa:2021qlt,Lotay:2021eog}, is an interesting avenue for future work.

As we discussed already in \cite{delaOssa:2019cci}, one motivation for our study of heterotic superpotentials is that these functionals allow us to explore finite deformations of heterotic systems. This has been established for the Hull-Strominger system in \cite{Ashmore:2018ybe}: by choosing a parametrisation for finite deformations suggested by holomorphy, it was shown that the superpotential expansion truncates at cubic order, and that the deformation algebra is an $L_3$ algebra. In \cite{delaOssa:2019cci} we conjectured that there should be a similar convenient parametrisation for the finite deformations of heterotic $G_2$ systems, which should lead to a superpotential expansion that truncates at fourth order. While we have not been able to confirm this suspicion in this paper, we have taken an important step towards its resolution by identifying a field redefinition that diagonalises the differential appearing in the second order variation of the superpotential. 

By interpreting the second order variation of the superpotential as a classical quadratic action in $\cal Z$, we have initiated a study of quantum aspects of the heterotic $G_2$ system. Using the formal analogy of the action with three-dimensional Chern-Simons theory, we could express the absolute value of the one-loop partition function in terms of determinants of Laplacians $\Delta^p_\newD$. This result has the same functional form as obtained in \cite{Kupka:2024rvl} using generalised geometry techniques. With our explicit double complex in hand, we can go further: we have shown that we may perform a field redefinition that diagonalizes $\newD$. Using this, the one-loop partition function can be rewritten in terms of more familiar Laplacians. This is along the lines of what is obtained from quantization of type II strings on $G_2$ spaces \cite{deBoer:2005pt, deBoer:2007zu, Ashmore:2021pdm}.

In this paper, we have not performed the full BV quantisation of the heterotic $G_2$ systems. This is a natural next step to take. We expect that, by using the double complex in its diagonalised format, it will be relatively straight forward to compute the phase of the one-loop partition function (i.e. the $\eta$ invariant). We may then compare with the results obtained from topological $G_2$ strings \cite{deBoer:2006bp,deBoer:2007zu}. Since, in contrast to these topological studies, our bicomplex encodes both gauge and geometric degrees of freedom, we may be able to provide new insights to possible geometric anomalies of the heterotic $G_2$ system. The intricate coupling between these degrees of freedom may also help expose to what extent topological invariants and enumerative geometry may be defined  for instantons in the $G_2$ context \cite{donaldson2009gauge, Joyce:2016fij}.

\section*{Acknowledgements}
We would like to thank Mario Garcia-Fernandez, Spiro Karigiannis, Julian Kupka, Sebastien Picard, Jock McOrist, Martin Sticka, Charles Strickland-Constable, David Tennyson, Fridrich Valach, for enlightening discussions.  ML and MM were supported in part by Vetenskapsrådet, grant no.~2020-03230 and grant no. 2023-00508, respectively. ML thanks the Erwin Schrödinger Institute in Vienna, and the organisers of the  program {\it The Landscape vs~the Swampland}, for hospitality during some parts of this work. XD and ES would like to thank the organisers of the scientific programmes \emph{The Geometry of Moduli Spaces in String Theory} and \emph{New Deformations of Quantum Field and Gravity Theories} respectively, hosted by the MATRIX Institute, for the hospitality during later stages of this work. XD thanks The Sydney Mathematical Institute for support and hospitality while this work was being completed. For the purpose of open access, the authors have applied a CC-BY public copyright license to any Author Accepted Manuscript (AAM) version arising from this submission.

\appendix
\section{Proof of proposition \ref{prop:del2W}}
\label{proof-prop:del2W}

\medskip

In this appendix we give a detailed proof of proposition \ref{prop:del2W}.  We begin by varying \eqref{eq:delwbis} which gives
\begin{align}
\begin{split}
\delta_{2}\delta_{1}{{\mathfrak{w}}} |_{0} &=  
 i_{\mathring M_{1}}(\psi)\wedge 
 \delta_{2}\left(H -\widehat T 
 \right)
\\[5pt]
&\quad  
+\frac{\alpha'}{2}\,\Big( \tr(\alpha_{1}\wedge \delta_{2}(F\wedge\psi)) - \tr(\kappa_{1}\wedge \delta_{2}(R\wedge\psi)) \Big)
\\[5pt]
&\quad - ({\cal B}_{1} - 2m_{1}) \wedge \delta_{2}\big( \dd\psi -2\,\dd\phi\wedge\psi\big)
\\[5pt]
&\quad + \delta_{2}\left(\delta_{1} h + \frac{3}{7}\,\tr M_{1}\, \Big(h + \frac{1}{3}\,\tau_{0}\Big) \right) \,   \varphi\wedge\psi  ~.
\end{split}\label{eq:deldelwbis}
\end{align}
We now analyse each term. 
\bigskip

\noindent\underline{First term of equation \eqref{eq:deldelwbis}}

\begin{align*}
 &i_{\mathring M_{1}}(\psi)\wedge 
 \delta_{2}\left(H - \widehat T 
 \right)
\\[5pt]
&\quad
=    i_{\mathring M_{1}}(\psi)\wedge
\left(\dd{\cal B}_{2}  +  \frac{\alpha'}{2}\, \Big( \tr (\alpha_{2} \wedge F) - \tr(\kappa_{2}\wedge R)\Big)
- \delta_{2}\widehat T 
\right) ~.
\end{align*}
To continue we use the following identity: 
\be
i_{M}(\psi)\wedge\lambda =   \tr M\, \lambda \wedge\psi - i_{M}(\lambda)\wedge\psi  ~,
\qquad\forall\lambda\in\Omega^{3}(Y)
\label{eq:lemma3form}
\ee
Then, noting that $\mathring M$ is traceless, we find (on shell)
\be\begin{split}
 i_{\mathring M_{1}}(\psi)\wedge 
 \delta_{2}\Big(&H - \widehat T 
 \Big)
\\[5pt]
&~ =   -  i_{\mathring M_{1}}
\Big[\dd{\cal B}_{2}  +  \frac{\alpha'}{2}\, \Big( \tr (\alpha_{2} \wedge F) - \tr(\kappa_{2}\wedge R)\Big) - \delta_{2}\widehat T
\Big] \wedge\psi~. \label{eq:termone}
\end{split}\ee

We need to express $\delta \widehat T$ in terms of $M$ by varying the structure equations \eqref{eq:dphi} and \eqref{eq:dpsi}
 being careful to only set $\tau_{2}= 0$  on shell while retainnig $\delta\tau_{2}$.
 This computation we done in section 3  of \cite{delaOssa:2017pqy} where the authors found 
\be
\delta\widehat T_{a}\lrcorner\varphi=  \big( 2\, (\dd m)_{a} + 2\, \dd_{\zeta} M_{a} \big)\lrcorner\varphi   + \delta\tau_{2\, ab}\,\dd x^{b}~,
\qquad \pi_{\bf 7}(\delta\tau_{2}) = 0~. \label{eq:pi7deltaT}
\ee
Since only $\pi_{\7}\big((\delta T)_{a}\big)$ contributes to \eqref{eq:termone}
we have
\be\begin{split}
  & i_{\mathring M_{1}}\left( \dd{\cal B}_{2} - \delta_{2} \widehat T\right)\wedge\psi
   =  i_{\mathring M_{1}}
     \Big(
    \dd(\pi_{\bf{7}}\big({\cal B}_{2})- 2 m_{2}\big) +  \dd_{\zeta} \big(2\, M_{2} + \pi_{\bf{14}}({\cal B}_{2})\big)
    \Big) \wedge\psi
    \\[8pt]
  &\qquad\qquad\qquad\qquad\qquad 
    + \mathring M_{1}^{a}\wedge *\big((\delta_{2}\tau_{2})_{ab}\dd x^{b}\big)  
   ~,
      \end{split}\label{eq:fortermone}
      \ee
 where we have used the identity
  \[
 \big( \dd_{\zeta} z_{a} + (\dd z)_{a}\big)\wedge\psi = 0~, \qquad \forall\, z\in \Omega^{2}_{\bf{14}}(Y)~.
  \]    
  As $\pi_{\bf 7}(\delta\tau_{2})= 0$, the last term in \eqref{eq:fortermone} vanishes
 \[
 \mathring M_{1}^{a}\wedge *\big((\delta_{2}\tau_{2})_{ab}\dd x^{b}\big)
 = 2\, *(m_{1} \lrcorner \delta_{2}\tau_{2})=0~.
 \]
   Then \eqref{eq:termone} becomes
 \be\begin{split}
 i_{\mathring M_{1}}(\psi)&\wedge 
 \delta_{2}\left( H - \widehat T 
 \right)
\\[8pt]
& =   2\, i_{\mathring M_{1}}
\left(
\dd_{\zeta} y_{2} 
-  \frac{\alpha'}{4}\, \Big( \tr (\alpha_{2} \wedge F) - \tr(\kappa_{2}\wedge R)\Big)
- \dd\beta_{2}
\right) \wedge\psi~,\label{eq:termonebis}
\end{split}\ee
where $\beta$  and $y$ are  given in \eqref{eq:beta} and \eqref{eq:y}.
\bigskip

\noindent\underline{Second term of equation \eqref{eq:deldelwbis}:}
Recall that \cite{atiyah57} 
\[
\delta F = \dd_{A}\alpha - i_{M}(F)~,
\]
with a similar equation for $\delta R$.  Then, on shell we obtain
\bigskip
  \begin{align*}
\delta_{2}&\left(\frac{\alpha'}{2}\,\Big( \tr(\alpha_{1}\wedge F) - \tr(\kappa_{1}\wedge R) \Big)\wedge \psi \right)
\\[5pt]
&\qquad\quad=  
\frac{\alpha'}{2}\,\Big( \tr\big(\alpha_{1}\wedge (\dd_{A}\alpha_{2} - i_{ M_{2}}(F))\big) 
- \tr\big(\kappa_{1}\wedge (\dd_{\Theta}\kappa_{2}-  i_{ M_{2}}(R) )\big)
 \Big)\wedge \psi
\end{align*}

\bigskip
\noindent\underline{Third term of equation \eqref{eq:deldelwbis}}
Recall \cite{delaOssa:2017pqy}
\[
\delta\dd\psi = i_{\dd_{\zeta}(M)}(\psi) + i_{M}(\dd \psi) = i_{\dd_{\zeta}(M)}(\psi) + i_{M}(4\tau_{1})\wedge\psi + 4\tau_{1}\wedge i_{M}(\psi)~.
\]
 Then
\begin{align*}
&({\cal B}_{1} - 2 m_{1})\wedge \delta_{2}(\dd\psi - 2\, \dd\phi\wedge\psi)
\\[5pt]
&~~ =  ({\cal B}_{1} - 2 m_{1})\wedge 
\big( i_{\dd_{\zeta}M_{2}}(\psi) + \big( i_{M_{2}} (4\, \tau_{1}) -  2\,\dd\delta_{2}\phi\big)\wedge\psi 
+ ( 4\, \tau_{1}-  2\,\dd\phi)\wedge i_{M_{2}}(\psi) \big)~.
\end{align*}
The last term vanishes on shell.  Then, separating $\cal B$ into its $\bf 7$ and $\bf 14$ parts we have
\be
\begin{split}
({\cal B}_{1} - 2 m_{1})&\wedge \delta_{2}(\dd\psi - 2\, \dd\phi\wedge\psi)
=   \pi_{\bf 14}({\cal B}_{1})\wedge i_{\dd_{\zeta}M_{2}}(\psi)
\\[5pt]
&\qquad 
+ 2\, \beta_{1}\wedge\left(
\frac{1}{3}\, \psi\lrcorner\,\left(i_{\dd_{\zeta}M_{2}}(\psi)\right) +  i_{M_{2}} (4\, \tau_{1}) -  2\,\delta_{2}(\dd\phi) \right)\wedge\psi
~.
\end{split}\label{eq:pretermtwo}
\ee

\noindent The first term of \eqref{eq:pretermtwo} can written as 
\begin{align*}
 \pi_{\bf 14}({\cal B}_{1}) \wedge i_{\dd_{\zeta} M_{2}}(\psi)&= 
 -  \pi_{\bf 14}({\cal B}_{1})^{a} \wedge i_{\dd_{\zeta} M_{2\, a}}\wedge\psi
 \\[5pt]
 &=-  \left( i_{\pi_{\bf 14}({\cal B}_{1})}(\dd_{\zeta} y_{2})
 - \frac{1}{2}\, i_{\pi_{\bf 14}({\cal B}_{1})}\big(\dd_{\zeta}\pi_{\bf 14}({\cal B}_{2})\big)
 \right)\wedge \psi
 \\[5pt]
 &= -\left( i_{\pi_{\bf 14}({\cal B}_{1})}(\dd_{\zeta} y_{2})
 - \frac{1}{2}\, i_{\pi_{\bf 14}({\cal B}_{1})}\big(\dd\pi_{\bf 14}({\cal B}_{2})\big)
 \right)\wedge \psi~,
\end{align*}
where the last term vanishes. 
Then the third term of \eqref{eq:deldelwbis} is given by
\be
\begin{split}
&- ({\cal B}_{1} - 2 m_{1})\wedge \delta_{2}(\dd\psi - 2\, \dd\phi\wedge\psi)
\\[5pt]
&=  
\left( i_{\pi_{\bf 14}({\cal B}_{1})}(\dd_{\zeta} y_{2})
- 2\, \beta_{1}\wedge\left(
\frac{1}{3}\, \psi\lrcorner\,\left(i_{\dd_{\zeta}M_{2}}(\psi)\right) +  i_{M_{2}} (4\, \tau_{1}) -  2\,\delta_{2}(\dd\phi) \right)\right)\wedge\psi
~.
\end{split}\ee

\bigskip

\noindent\underline{Fourth term of equation \eqref{eq:deldelwbis}:}
On shell, this term gives
\be
\delta_{2}\left(\delta_{1} h + \frac{3}{7}\,\tr M_{1}\, \Big(h + \frac{1}{3}\,\tau_{0}\Big) \right) \,   \varphi\wedge\psi 
=  \frac{1}{7}\,\tr M_{1}\,  \delta_{2}\tau_{0} \,   \varphi\wedge\psi 
\ee
To compute $\delta\tau_{0}$ in terms of $M$, recall
\[
\tau_{0} = \frac{1}{7}\, \dd \varphi\lrcorner\psi~.
\]
Taking the Hodge-dual
\[
\tau_{0}\, \varphi\wedge\psi = \dd\varphi\wedge\varphi~,
\]
and varying this equation we find
\[
\begin{split}
(\delta \tau_{0} + \tau_{0}\,\tr M)\varphi\wedge\psi 
&= \dd i_{M}(\varphi)\wedge\varphi + \dd\varphi\wedge i_{M}(\varphi)
=i_{\dd_{\zeta} M}(\varphi)\wedge\varphi + i_{M}(\dd\varphi\wedge\varphi)
\\[5pt]
&= i_{\dd_{\zeta} M}(\varphi)\wedge\varphi + \tau_{0}\, \tr M \varphi\wedge\psi~.
\end{split}
\]
Then
\[
\delta \tau_{0}= \frac{1}{7}\, \psi\lrcorner i_{\dd_{\zeta} M}(\varphi)~.
\]
\bigskip
It is not hard to prove that
\be
\delta \tau_{0}= \frac{1}{7}\, \psi\lrcorner i_{\dd_{\zeta} M}(\varphi) = 
\frac{2}{7}\,  \big(-2\,\dd m + i_{M}(\widehat T)\big)\lrcorner\varphi
~,\label{eq:deltatau0}
\ee
which implies
\[
\delta \tau_{0}\,\varphi\wedge\psi =  2 \big(-2\, \dd m + i_{M}(\widehat T)\big)\wedge\psi~.
\]
Then, the fourth term of \eqref{eq:deldelwbis} becomes
\be
\delta_{2}\left(\delta_{1} h + \frac{3}{7}\,\tr M_{1}\, \Big(h + \frac{1}{3}\,\tau_{0}\Big) \right) \,   \varphi\wedge\psi 
=  \frac{2}{7}\,\tr M_{1}\, \big(-2\,\dd m_{2} + i_{M_{2}}(T)\big)\wedge\psi~.
\ee

Putting together the results we have so far, we find
  \be
\begin{split}
\frac{1}{2}\,\delta_{2}\delta_{1}{{\mathfrak{w}}}\,|_{0} 
& =   \Big\{ i_{\mathring M_{1}}
\Big[
~\dd_{\zeta} y_{2} 
-  \frac{\alpha'}{4}\, \Big( \tr (\alpha_{2} \wedge F) - \tr(\kappa_{2}\wedge R)\Big) 
\Big]
- i_{\mathring M_{1}}(\dd\beta_{2})
\\[5pt]
&\quad
+ \, \frac{1}{7}\,\tr M_{1}\, \left(-2\,\dd m_{2} + i_{M_{2}}(T) \right)
+\frac{1}{2}\, i_{\pi_{\bf 14}({\cal B}_{1})}(\dd_{\zeta} y_{2}) \\[5pt]
 &\quad+  
\frac{\alpha'}{4}\,\Big[ \tr\big(\alpha_{1}\wedge (\dd_{A}\alpha_{2} - i_{ M_{2}}(F))\big) 
- \tr\big(\kappa_{1}\wedge (\dd_{\Theta}\kappa_{2}-  i_{ M_{2}}(R) )\big)
 \Big]
 \\[5pt]
& \quad 
-  \beta_{1}\wedge\Big[\,
\frac{1}{3}\, \psi\lrcorner\,\left(i_{\dd_{\zeta}M_{2}}(\psi)\right) +  i_{M_{2}} (4\, \tau_{1}) -  2\, \delta_{2}(\dd\phi) \Big]
\Big\}\wedge\psi~.
 \end{split}\label{eq:sofar}
 \ee
 We need to somewhat reorganize the terms to obtain the equation \eqref{eq:deldelw} claimed in the proposition.
 Note that only the third term in \eqref{eq:sofar} contains $\tr M_{1}$.  This term precisely combines with the first term giving
 \[
 \begin{split}
 &\Big\{ i_{\mathring M_{1}}
\Big[
~\dd_{\zeta} y_{2} 
-  \frac{\alpha'}{4}\, \Big( \tr (\alpha_{2} \wedge F) - \tr(\kappa_{2}\wedge R)\Big) 
\Big]
+ \, \frac{1}{7}\,\tr M_{1}\, \left(-2\dd m_{2} + i_{M_{2}}(T)  \right)
\Big\}\wedge\psi
\\[8pt]
& =   i_{M_{1}}
\Big[
 \dd_{\zeta} y_{2} 
-  \frac{\alpha'}{4}\, \Big( \tr (\alpha_{2} \wedge F) - \tr(\kappa_{2}\wedge R)\Big) 
\Big]\wedge\psi
\end{split}
\]  
 The computation to see this is straightforward.  
  \begin{align*}
  i_{\frac{1}{7}\, \tr M_{1}{\bf 1}}\big(\dd_{\zeta} y_{2}\big)\wedge\psi&= \frac{1}{7}\, \tr M_{1}\, \dd x^{a}\wedge\dd_{\zeta} y_{2\, a}\wedge\psi
  = \frac{1}{7}\, \tr M_{1}\, \dd x^{a}\wedge\dd_{\zeta} M_{2\, a}\wedge\psi~,
  \end{align*}
  where $\pi_{\bf{14}}(\cal B)$ drops out because of the product with $\psi$:
  \[
  \dd x^{a}\wedge\dd_{\zeta}\big(\pi_{\bf{14}}({\cal B}))_{a} \big)\wedge\psi
 =  - \dd x^{a}\wedge\dd \big(\pi_{\bf{14}}({\cal B}) \big)_{a}\wedge\psi = - \dd \big(\pi_{\bf{14}}({\cal B}) \big)\wedge\psi =0~,
  \]
  as $\pi_{\bf 1}\dd(\pi_{\bf{14}}({\cal B}) = 0$ when $\tau_{2}= 0$.
  Then
  \begin{align*}
  i_{\frac{1}{7}\, \tr M_{1}{\bf 1}}\big(\dd_{\zeta} y_{2}\big)\wedge\psi
  &= \frac{2}{7}\, \tr M_{1}\, \left(\partial_{b} m_{2\, ac} + 
  \frac{1}{2}\, \hat T^{e}{}_{ab}\, M_{2\, ec}\right)\, \dd x^{abc}\wedge\psi
  \\[8pt]
 & = \frac{1}{7}\, \tr M_{1}\, \left(- 2\, \dd m_{2} +  i_{M_{2}}(\hat T)\right)\, \wedge\psi
  ~.
  \end{align*} 
Note also that the third term of equation \eqref{eq:sofar} can also be combined with the first. 
Then, we have
 \be
\begin{split}
 \frac{1}{2}\,\delta_{2}\delta_{1}{{\mathfrak{w}}}\,|_{0} 
& =  \Big\{i_{y_{1}}
\Big[
\dd_{\zeta} y_{2} 
-   \frac{\alpha'}{4}\, \Big( \tr (\alpha_{2} \wedge F) - \tr(\kappa_{2}\wedge R)\Big) 
\Big]
- i_{\mathring M_{1}}(\dd\beta_{2})
\\[5pt]
 &\quad+  
\frac{\alpha'}{4}\,\Big[ \tr\big(\alpha_{1}\wedge (\dd_{A}\alpha_{2} - i_{ M_{2}}(F))\big) 
- \tr\big(\kappa_{1}\wedge (\dd_{\Theta}\kappa_{2}-  i_{ M_{2}}(R) )\big)
 \Big]
 \\[5pt]
&\quad 
-  \beta_{1}\wedge\Big[\,
\frac{1}{3}\, \psi\lrcorner\,\left(i_{\dd_{\zeta}M_{2}}(\psi)\right) +  i_{M_{2}} (4\, \tau_{1}) -  2\,\delta_{2}(\dd\phi) \Big]
\Big\}\wedge\psi
~.
 \end{split}
 \ee
 Finally, we observe that the last term can be written as 
 \begin{equation}
\begin{split}
 \beta_{1}\wedge&\left[\,
\frac{1}{3}\, \psi\lrcorner\,\left(i_{\dd_{\zeta}M_{2}}(\psi)\right) +  i_{M_{2}} (4\, \tau_{1}) -  2\,\dd\delta_{2}\phi \right]
 \wedge\psi 
 \\[8pt]
 &\qquad\qquad\qquad
 =  \beta_{1}\wedge\left[ \, 
\frac{1}{3}\, \psi\lrcorner\,\left(i_{\dd_{\zeta}{\mathring M}_{2}}(\psi)\right) +  i_{{\mathring M}_{2}} (4\, \tau_{1}) 
-  2\,\dd\hat\delta_{2}\phi \right]
 \wedge\psi 
 ~,
\end{split}
  \end{equation}
  where we have defined
  \[
  \hat\delta\phi =  \delta\phi - \frac{2}{7}\, \tr M~.
  \]
  In passing from the first line to the second in this equation we have used
  \[
\begin{split}
i_{\dd_{\zeta}(\tr M\, {\bf 1})}(\psi) &= \big(\dd(\tr M)\wedge\dd x^{a} -  \tr M\,T^{a}\big)\wedge\psi_{a}  
= 4\, \dd(\tr M) \wedge\psi - i_{T}(\psi) 
\\[5pt]
&= 4\,\big(\dd(\tr M) - \tr M\, \tau_{1}\big)\wedge\psi
~.
\end{split}
\]
This proves the proposition. 
 \bigskip
 $\null \hfill\square$ 
\bigskip

\section{Ellipticity}
\label{sec:elliptic}

In this appendix, we prove that the $(\cal{E},\newD)$ complex is elliptic. Good introductions to the notions of  ellipticity of operators and complexes,
can be found in the book of Gilkey \cite{Gilkey}, and the lecture notes of Pati \cite{Pati}. A very brief summary of most of the notions we need can also be found in appendix B of \cite{delaOssa:2016ivz}.

Recall first that a complex $(E,D)$ is elliptic if the induced complex $(E,\sigma_L(D))$ of leading (or principal) symbols $\sigma_L(D)$ is exact. Reyes-Carrion \cite{ReyesCarrion:1998si,ReyesCarrionPhd} proved this is the case for the canonical $G_2$ complex with differential $\check{\dd}$, and in \cite{delaOssa:2016ivz,delaOssa:2017pqy} this was used to prove ellipticity of the $\check{\cal{D}}$-complex. 

Second, we will make use of
\begin{lemma} \label{lemma-ellipt}
	Let $E:=\bigoplus L_i,\,F:=\bigoplus M_j$ two vector bundles, given by a direct sum with the same index set.
	Let $D:E\rightarrow F$ a differential operator such that, with respect to the above decomposition: $D$ is upper triangular; all entries, $D_{ji}:L_i\rightarrow M_j$, have the same degree; and, the diagonal entries $D_{ii}$ are elliptic.
	Then, the total operator, $D$, is elliptic.
\end{lemma}
\begin{proof}
	Given our assumptions, the symbol of $D$ is a matrix $\sigma_\xi(D)$ with entries\break \hbox{$\sigma_\xi(D)_{ij}=\sigma_\xi(D_{ij})$}, so the result follows by linear algebra.
    
\end{proof}

Clearly, $\newD$ is upper triangular, and  by the results quoted before the Lemma, the diagonal entries of $\newD$ are elliptic. Moreover, the off-diagonal elements are given by the $\cal{S}$ and $\cal{T}$ maps, which are all degree-1 operators. Thus, the conditions for Lemma \ref{lemma-ellipt} are fulfilled, and we conclude the $(\cal{E},\newD)$ complex is elliptic.

Lemma \ref{lemma-ellipt} on the level of the Laplacians 
gives the same result.

\bibliographystyle{JHEP.bst}
\bibliography{bibilio}
\end{document}